\documentclass[10pt]{article}

\usepackage[american]{babel} %
\usepackage[T1]{fontenc} %
\usepackage[latin1]{inputenc} %

\usepackage{enumerate}
\usepackage{amsfonts,amsmath,amssymb,amsthm,amstext,latexsym,paralist}
\usepackage{url}
\usepackage{xspace}
\usepackage[ruled,lined,boxed,commentsnumbered]{algorithm2e}
\usepackage{tikz}
\usepackage{subfig}
\usepackage{verbatim}
\usepackage{array}
\usepackage{fancyhdr}
\usepackage{graphicx}
\usepackage{wrapfig}
\usepackage{times}
\usepackage{longtable}
\usepackage{lscape}
\usepackage{fullpage}

\usepackage{fixltx2e}
\usepackage{stfloats}
\usepackage{mdwmath}
\usepackage{mdwtab}

\theoremstyle{plain}
\newtheorem{lemma}{Lemma}
\newtheorem{theorem}{Theorem}
\newtheorem{prop}{Proposition}

\theoremstyle{definition}

\theoremstyle{remark}

\newcommand{\continuous}{\textsc{Continuous}\xspace}
\newcommand{\VDD}{\textsc{Vdd-Hopping}\xspace}
\newcommand\II{\ensuremath{\mathcal{I}}\xspace}

\newcommand{\fmax}{\ensuremath{f_{\max}}\xspace}
\newcommand{\fmin}{\ensuremath{f_{\min}}\xspace}
\newcommand{\finf}{\ensuremath{f_{i}^{(\texttt{inf})}}\xspace}
\newcommand{\finfc}{\ensuremath{f^{(\texttt{inf})}}\xspace}
\newcommand{\fr}{\ensuremath{f_{\texttt{rel}}}\xspace}
\newcommand{\fdec}{\ensuremath{f_{\texttt{dec}}}\xspace}
\newcommand{\freex}{\ensuremath{f_{\texttt{re-ex}}}\xspace}
\newcommand{\sss}{\ensuremath{f_{\texttt{src}}}\xspace}
\newcommand{\ttt}{\ensuremath{f_{\texttt{leaf}}}\xspace}

\newcommand{\DDmin}{\ensuremath{D_{\texttt{min}}}\xspace}
\newcommand{\DDratio}{\ensuremath{\textsc{DeadlineRatio}}\xspace}
\newcommand{\exe}{\ensuremath{\mathcal{E}\!xe}\xspace}

\newcommand{\tricritcont}{\textsc{Tri-Crit-Cont}\xspace}
\newcommand{\tricritvdd}{\textsc{Tri-Crit-Vdd}\xspace}
\newcommand{\chaincont}{\textsc{Tri-Crit-Cont-Chain}\xspace}
\newcommand{\chainvdd}{\textsc{Tri-Crit-Vdd-Chain}\xspace}
\newcommand{\forkcont}{\textsc{Tri-Crit-Cont-Fork}\xspace}

\graphicspath{{images/}}

\begin{document}

\title{Energy-aware scheduling under reliability 
and makespan constraints}

\author{Guillaume Aupy\thanks{Ecole Normale Sup\'erieure de Lyon, France} \and Anne Benoit\footnotemark[1] \and Yves Robert\footnotemark[1]~\thanks{University of Tennessee Knoxville, USA} \\ 
\{guillaume.aupy, anne.benoit, yves.robert\}@ens-lyon.fr}

\maketitle

\begin{abstract}
   We consider a task graph mapped on a set of homogeneous
  processors. We aim at minimizing the energy consumption while
  enforcing two constraints: a prescribed bound on the execution time
  (or makespan), and a reliability threshold.  Dynamic voltage and
  frequency scaling (DVFS) is an approach frequently used to reduce
  the energy consumption of a schedule, but slowing down the execution
  of a task to save energy is decreasing the reliability of the execution.
  In this work, to improve the reliability of a
  schedule while reducing the energy consumption, we allow for the
  re-execution of some tasks.  We assess the complexity of the
  tri-criteria scheduling problem (makespan, reliability, energy) of
  deciding which task to re-execute, and at which speed each execution
  of a task should be done, with two different speed models: either
  processors can have arbitrary speeds (\continuous model), or a
  processor can run at a finite number of different speeds and 
  change its speed during a computation (\VDD model). We propose
  several novel tri-criteria scheduling heuristics under the
  continuous speed model, and we evaluate them through a set of
  simulations. The two best heuristics turn out to be very efficient
  and complementary. 
\end{abstract}

\clearpage

\section{Introduction}

Energy-aware scheduling has proven an important issue in the past
decade, both for economical and environmental reasons. This holds true
for traditional computer systems, not even to speak of battery-powered
systems. More precisely, a processor running at speed $s$ dissipates
$s^3$ watts per unit of time~\cite{pow3IPDPS,BKP07,pow3ICPP}, hence it
consumes $s^3 \times d$ joules when operated during $d$ units of
time. To help reduce energy dissipation, processors can run at
different speeds. A widely used technique to reduce energy consumption
is \emph{dynamic voltage and frequency scaling (DVFS)}, also known as
speed scaling~\cite{pow3IPDPS,BKP07,pow3ICPP}. Indeed, by lowering
supply voltage, hence processor clock frequency, it is possible to
achieve important reductions in power consumption; 
faster speeds allow for a faster execution, but they also lead to a
much higher (supra-linear) power consumption. 
There are two popular models for processor speeds. In the \continuous
model, processors can have arbitrary speeds, and can vary them
continuously in the interval $[\fmin,\fmax]$. This model is
unrealistic (any possible value of the speed, say $\sqrt{e^{^\pi}}$,
cannot be obtained), but it is theoretically appealing~\cite{BKP07}.
In the \VDD model, a processor can run at a finite number of different
speeds ($f_1,...,f_m$). It can also change its speed during a
computation ({\em hopping} between different voltages, and hence
speeds). Any rational speed can therefore be
simulated~\cite{Miermont2007}. The energy consumed during the
execution of one task is the sum, on each time interval with constant
speed~$f$, of the energy consumed during this interval at speed~$f$.

Energy-aware scheduling aims at minimizing the energy consumed during
the execution of the target application. Obviously, this goal makes sense
only when coupled with some performance bound to achieve,
otherwise, the optimal solution always is to run each processor at the
slowest possible speed. In this paper, we consider a directed acyclic
graph (DAG) of $n$~tasks with precedence constraints, and the goal is
to schedule such an application onto a fully homogeneous platform
consisting of $p$~identical processors. This problem has been widely
studied with the objective of minimizing the total execution time, or
{\em makespan}, and it is well known to be
NP-complete~\cite{Brucker}. 
Since the introduction of
DVFS, many papers have dealt with the optimization of energy
consumption while enforcing a deadline, i.e., a bound on the
makespan~\cite{pow3IPDPS,BKP07,pow3ICPP,aupy12ccpe}.

There are many situations in which the mapping of the task graph is
given, say by an ordered list of tasks to execute on each processor,
and we do not have the freedom to change the assignment of a given
task. Such a problem occurs when optimizing for legacy applications,
or accounting for affinities between tasks and resources, or even when
tasks are pre-allocated~\cite{Rayward95}, for example for security
reasons. 
While it is not possible to change the allocation of a task, it is
possible to change its speed. This technique, which consists in
exploiting the slack due to workload variations, is called slack
reclaiming~\cite{LeeSakurai00,prathipati2004}. 
In our previous work~\cite{aupy12ccpe}, 
assuming that the mapping and a deadline are given, we have assessed
the impact of several speed variation models on the complexity of the
problem of minimizing the energy consumption. 
Rather than using a local approach such as
backfilling~\cite{Wang2010,prathipati2004}, which only reclaims gaps
in the schedule, we have considered the problem as a whole.

While energy consumption can be reduced by using speed scaling
techniques, it was shown in \cite{Zhu04EEM,Degal05SEI} that reducing
the speed of a processor increases the number of transient fault rates
of the system; the probability of failures increases exponentially,
and this probability cannot be neglected in large-scale
computing~\cite{Oliner04}. 
In order to make up for the loss in {\em reliability} due to the energy
efficiency, different models have been proposed for fault-tolerance: 
(i) \emph{re-execution} is the model under study in this work,
  and it consists in re-executing a task that does not meet the
  reliability constraint; it was also studied in
  \cite{Zhu04EEM,Zhu06,Izo07}; 
  (ii) \emph{replication} was studied in \cite{Assayad11,Girault09};
  this model consists in executing the same task on several processors
  simultaneously, in order to meet the reliability constraints; 
and (iii) \emph{checkpointing} 
  consists in "saving" the work done at
  some certain points of the work, hence reducing the amount of work
  lost when a failure occurs~\cite{Melhem03CP,Zhang03CP}.

This work focuses on the re-execution model, for several reasons. On
the one hand, replication is too costly in terms of both resource
usage and energy consumption: even if the first execution turns out
successful (no failure occurred), the other executions will still have
to take place.  Moreover, the decision of which tasks should be
replicated cannot be taken when the mapping is already fixed. On the
other hand, checkpointing is hard to manage with parallel processors,
and too costly if there are not too many failures. Altogether, it is
the "online/no-waste" characteristic of the corresponding algorithms
that lead us focus on re-execution.  The goal is then to ensure that
each task is reliable enough, i.e., either its execution speed is
above a threshold, ensuring a given reliability of the task, or the
task is executed twice to enhance its reliability. There is a clear
trade-off between energy consumption and reliability, since decreasing
the execution speed of a task, and hence the corresponding energy
consumption, is deteriorating the reliability.  This calls for
tackling the problem of considering the three criteria (makespan,
reliability, energy) simultaneously. This tri-criteria optimization
brings dramatic complications: in addition to choosing the speed of
each task, as in the deadline/energy bi-criteria problem, we also need
to decide which subset of tasks should be re-executed (and then choose
both execution speeds).  Few authors have tackled this problem; we detail below
the closest works to ours~\cite{Izo07,Zhu06,Assayad11}.

Izosinov et al. \cite{Izo07} study a tri-criteria optimization problem
with a given mapping on heterogeneous architectures.  However, they do
not have any formal energy model, and they assume that the user
specifies the maximum number of failures per processor tolerated to
satisfy the reliability constraint, while we consider any number of
failures but ensure a reliability threshold for each task.  
Zhu and Aydin~\cite{Zhu06}  
are also addressing a tri-criteria
optimization problem similar to ours, and choose some tasks that have
to be re-executed to match the reliability constraint. However, they
restrict to the scheduling problem on one single processor, and they
consider only the energy consumption of the first execution of a task
(best-case scenario) when re-execution is done.  
Finally, Assayad et al. \cite{Assayad11} have recently proposed an
off-line tri-criteria scheduling heuristic (TSH), which uses active
replication to minimize the makespan, with a threshold on the global
failure rate and the maximum power consumption.
TSH is an improved critical-path list scheduling heuristic that takes
into account power and reliability before deciding which task to
assign and to duplicate onto the next free processors. The complexity
of this heuristic is unfortunately exponential in the number of
processors.
Future work will be devoted to compare our heuristics to TSH, and
hence to compare re-execution with replication.

Given an application with dependence constraints and a mapping of this
application on a homogeneous platform, we present in this paper
theoretical results and tri-criteria heuristics that use re-execution
in order to minimize the energy consumption under the constraints of
both a reliability threshold per task and a deadline
bound. 
The first contribution is a formal model for this tri-criteria
scheduling problem
(Section~\ref{models}). The second contribution is to provide
theoretical results for the different speed models, \continuous
(Section~\ref{cont_model}) and \VDD (Section~\ref{vdd_model}).  
The third contribution is the design of novel tri-criteria scheduling
heuristics that use re-execution to increase the reliability of a
system under the \continuous model (Section~\ref{heurWC}), and their
evaluation through extensive simulations
(Section~\ref{sec.exp_results}).  To the best of our knowledge, this
work is the first attempt to propose practical solutions to this
tri-criteria problem.  Finally, we give concluding remarks and
directions for future work in Section~\ref{sec.conclusion}.

\section{The tri-criteria problem}
\label{models}
	
Consider an application task graph $\mathcal{G}=(V,\mathcal{E})$,
where $V= \{T_1, T_2, \dots, T_n\}$ is the set of tasks, $n = |V|$,
and where $\mathcal{E}$ is the set of precedence edges between
tasks. For $1 \leq i \leq n$, task~$T_i$ has a weight~$w_i$, that
corresponds to the computation requirement of the task.  
We also consider particular class of task graphs, such as {\em linear
  chains} where $\mathcal{E}=\cup_{i=1}^{n-1}\{T_i \rightarrow
T_{i+1}\}$, and {\em forks} with $n+1$ tasks $\{T_0,T_1, T_2, \dots,
T_n\}$ and $\mathcal{E}=\cup_{i=1}^{n}\{T_0 \rightarrow T_{i}\}$.

We assume that tasks are mapped onto a parallel platform made up of
$p$ identical processors.  Each processor has a set of available
speeds that is either continuous (in the interval $[\fmin,\fmax]$) or
discrete (with $m$ modes $\{f_1,\cdots,f_m\}$), depending on the speed
model (\continuous or \VDD).  The goal is to minimize the energy
consumed during the execution of the graph while enforcing a deadline
bound and matching a reliability threshold.  To match the reliability
threshold, some tasks are executed once at a speed high enough to
satisfy the constraint, while some other tasks need to be
re-executed. We detail below the conditions that are enforced on the
corresponding execution speeds. The problem is therefore to decide
which task to re-execute, and at which speed to run each execution of
a task. 

In this section, for the sake of clarity, we assume that a task is
executed at the same (unique) speed throughout execution, or at two
different speeds in the case of re-execution.  In
Section~\ref{cont_model}, we show that this strategy is indeed optimal
for the \continuous model; in Section~\ref{vdd_model}, we show that
only two different speeds are needed for the \VDD model (and we update
the corresponding formulas accordingly). We now detail the three
objective criteria (makespan, reliability, energy), and then define
formally the problem. 

\subsection{Makespan}
\label{makespan_models}

The makespan of a schedule is its total execution time.  The
first task is scheduled at time $0$, so that the makespan of a
schedule is simply the maximum time at which one of the processors
finishes its computations.  We consider a {\em deadline bound} $D$, which is
a constraint on the makespan. 

Let $\exe(w_i,f)$ be the execution time of a task $T_i$ of weight
$w_i$ at speed $f$. We assume that the cache size is adapted to the
application, therefore ensuring that the execution time is linearly
related to the frequency~\cite{Melhem03CP}: $\exe(w_i,f) =
\frac{w_i}{f}$.  When a task is scheduled to be re-executed at two
different speeds $f^{(1)}$ and~$f^{(2)}$, we always account for both
executions, even when the first execution is successful, and hence 
$\exe(w_i,f^{(1)},f^{(2)}) = \frac{w_i}{f^{(1)}} +
\frac{w_i}{f^{(2)}}$. In other words, we consider a worst-case
execution scenario, and the deadline~$D$ must be matched even in the
case where all tasks that are re-executed fail during their first
execution.

\subsection{Reliability}
\label{reliability_models}

To define the reliability, we use the fault model of Zhu et
al.~\cite{Zhu04EEM,Zhu06}. \emph{Transient} failures are faults caused
by software errors for example. They invalidate only the execution of
the current task and the processor subject to that failure will be
able to recover and execute the subsequent task assigned to it (if
any). In addition, we use the reliability model introduced by Shatz
and Wang \cite{Shatz89}, which states that the radiation-induced
transient faults follow a Poisson distribution.  The parameter
$\lambda$ of the Poisson distribution is then:\\ 
\begin{equation}
	\label{fault}
\lambda(f)=\tilde{\lambda_0} \;  e^{\tilde{d}\frac{\fmax-f}{\fmax-\fmin}},
\end{equation}
where $\fmin\leq f \leq \fmax$ is the processing speed, the exponent
$\tilde{d}\geq0$ is a constant, indicating the sensitivity of fault
rates to DVFS, and $\tilde{\lambda_0}$ is the average fault rate
corresponding to \fmax.  We see that reducing the speed for energy
saving increases the fault rate exponentially.  The reliability of a
task $T_i$ executed once at speed $f$ is
$R_i(f)=e^{-\lambda(f)\times\exe(w_i,f)}$.  
Because the fault rate is usually very small, of the order of
$10^{-6}$ per time unit in~\cite{Baleani03,Izo07}, $10^{-5}$
in~\cite{Assayad11}, we can use the first order approximation of
$R_i(f)$ as\\ 
\begin{equation}
	\label{rel_first_order}
R_i(f) = 1-\lambda(f)\times\exe(w_i,f) 
       = 1-\tilde{\lambda_0}\; e^{\tilde{d}\frac{\fmax-f}{\fmax-\fmin}} \times \frac{w_i}{f} 
       = 1-\lambda_0\;e^{-df} \times \frac{w_i}{f},
\end{equation}
where $d=\frac{\tilde{d}}{\fmax-\fmin}$ and $\lambda_0 =
\tilde{\lambda_0} e^{d\fmax}$. This equation holds if $\varepsilon_{i}
= \lambda(f) \times \frac{w_i}{f} \ll 1$.  With, say, 
$\lambda(f) = 10^{-5}$, we need $\frac{w_i}{f} \leq 10^{3}$ to get
an accurate approximation with $\varepsilon_{i} \leq 0.01$: the task
should execute within $16$ minutes. In other words, large
(computationally demanding) tasks require reasonably high processing
speeds with this model (which makes full sense in practice).

We want the reliability~$R_i$ of each task $T_{i}$ to be greater than a
given threshold, namely $R_{i}(\fr)$, hence enforcing a local
constraint dependent on the task $R_i \geq R_{i}(\fr)$.  
If task~$T_{i}$ is executed only
once at speed~$f$, then the reliability of~$T_i$ is
$R_i=R_i(f)$. Since the reliability increases with speed, we must have
$f\geq \fr$ to match the reliability constraint. 
If task $T_{i}$ is re-executed (speeds $f^{(1)}$ and~$f^{(2)}$), then
the execution of~$T_{i}$ is successful if and only if both attempts do
not fail, so that the reliability of~$T_{i}$ is $R_{i}= 1 - (1 -
R_i(f^{(1)}))( 1 - R_i(f^{(2)}))$, and this quantity should be at
least equal to $R_{i}(\fr)$.

\subsection{Energy}
\label{energy_models}

The total energy consumption corresponds to the sum of the energy
consumption of each task. Let $E_i$ be the energy consumed by
task~$T_i$. For one execution of task~$T_i$ at speed~$f$, the
corresponding energy consumption is $E_i(f) =\exe(w_i,f) \times f^3 =
w_i \times f^2$, 
which corresponds to the dynamic part of the classical energy  
models of the literature~\cite{pow3IPDPS,BKP07,pow3ICPP,aupy12ccpe}.
Note that we do not take static energy into account, because
all processors are up and alive during the whole execution.

If task~$T_{i}$ is executed only once at speed~$f$, then 
$E_{i} = E_i(f)$. 
Otherwise, if task~$T_{i}$ is re-executed at speeds $f^{(1)}$ and $f^{(2)}$, 
it is natural to add up the energy consumed during both executions,
just as we add up both execution times when enforcing the makespan
deadline. Again, this corresponds to the worst-case execution
scenario.  We obtain 
$E_i = E_i(f^{(1)}_i) + E_i(f^{(2)}_i)$.  
Note that some authors~\cite{Zhu06} consider only the energy spent for
the first execution, which seems unfair: re-execution comes at a price
both in the deadline and in the energy consumption.  
Finally, the total energy consumed by the schedule, which we aim at
minimizing, is 
$E = \sum_{i=1}^{n} E_i$.

\subsection{Optimization problems}
\label{opt_problem}

The two main optimization problems are derived from the two different
speed models:  

\begin{compactitem}
\item {\tricritcont. } Given an application graph
  $\mathcal{G}=(V,\mathcal{E})$, mapped onto $p$~homogeneous
  processors with continuous speeds, \tricritcont is the problem of
  deciding which tasks should be re-executed and at which speed each
  execution of a task should be processed, in order to minimize the
  total energy consumption~$E$,
  subject to the deadline bound~$D$ and to the local reliability
  constraints $R_i \geq R_i(\fr)$ for each $T_i \in
  V$. 

\item {\tricritvdd. }
  This is the same problem as \tricritcont, but with the \VDD 
  model.
\end{compactitem}

\medskip

We also introduce variants of the problems for particular application
graphs: \chaincont is the same problem as
\tricritcont when the task graph is a linear
chain, mapped on a single processor; 
and \forkcont is the same problem as
\tricritcont when the task graph is a fork, and
each task is mapped on a distinct processor. We have similar
definitions for the \VDD model.

\section{\continuous model}
\label{cont_model}

As stated in Section~\ref{models}, we start by proving 
that with the \continuous model, it is
always optimal to execute a task at a unique speed throughout its
execution:

\begin{lemma}
\label{single_speed}
With the \continuous model, it is optimal to execute each task at a 
unique speed throughout its execution. 
\end{lemma}

The idea is to consider a task whose speed changes during the
execution; we exhibit a speed such that the execution time of the task
remains the same, but where both energy and reliability are
potentially improved, by convexity of the functions.  

\begin{proof}
  We can assume without loss of generality that the function that
  gives the speed of the execution of a task is a piecewise-constant
  function. The proof of the general case is a direct corollary from
  the theorem that states that any piecewise-continuous function
  defined on an interval $[a,b]$ can be uniformly approximated as 
  closely as desired by a piecewise-constant 
  function~\cite{rudin-principles}.

  Suppose that in the optimal solution, there is a task whose speed
  changes during the execution. Consider the first time-step at which
  the change occurs: the computation begins at speed~$f$ from time~$t$
  to time~$t'$, and then continues at speed~$f'$ until time~$t''$.
  The total energy consumption for this task in the time
  interval~$[t,t'']$ is $E=(t'-t)\times f^3+(t''-t')\times (f')^3$.
  Moreover, the amount of work done for this task is $W=(t'-t)\times
  f+(t''-t')\times f'$.  The reliability of the task is exactly
  $1-\lambda_0\left( (t'-t)\times e^{-df}+(t''-t')\times e^{-df'} +
    r\right )$, where $r$ is a constant due to the reliability of the
  rest of the process, which is independent from what happens during
  $[t,t'']$. The reliability is a function that increases when the
  function $h(t,t',t'',f,f') = (t'-t)\times e^{-df}+(t''-t')\times
  e^{-df'}$ decreases. 

  If we run the task during the whole interval~$[t,t'']$ at constant
  speed~$f_d=W/(t''-t)$, the same amount of work is done within the same
  time, and the energy consumption during this interval of time
  becomes $E'=(t''-t)\times f_d^3$.  
  Note that the new speed can be expressed as $f_d = af+(1-a)f'$,
  where $0<a = \frac{t'-t}{t''-t}<1$.
  Therefore, because of the convexity of the function $x \mapsto x^3$,
  we have $E'<E$. Similarly, since $x \mapsto e^{-dx}$ is a convex
  function, $h(t,t',t'',f_d,f_d)<h(t,t',t'',f,f')$, and the
  reliability constraint is also matched.
  This contradicts the hypothesis of optimality of the first solution,
  and concludes the proof.
\end{proof}

Next we show that not only a task is executed at a single speed, but
that its re-execution (whenever it occurs) is executed at the same speed
as its first execution:

\begin{lemma}
\label{WC_re-ex}
With the \continuous model, it is optimal to re-execute each task
(whenever needed) at the same speed as its first execution, and this
speed $f$ is such that $\finf \leq  f< \frac{1}{\sqrt{2}}\fr$, where\\[-.3cm]
 \begin{equation}
 	\label{eq.fmin}
 \lambda_0 w_i \frac{e^{-2d\finf}}{(\finf)^{2}} = \frac{e^{-d\fr}}{\fr}.
 \end{equation} 
\end{lemma}

Similarly to the proof of Lemma~\ref{single_speed}, we exhibit a
unique speed for both executions, in case they differ, so that the
execution time remains identical but both energy and reliability are
improved. If this unique speed is greater than
$\frac{1}{\sqrt{2}}\fr$, then it is better to execute the task only
once at speed~$\fr$, and if $f$ is lower than~$\finf$, then the
reliability constraint is not matched. 

\begin{proof}
  Consider a task $T_i$ executed a first time at speed $f_i$, and a
  second time at speed $f'_i>f_i$. Assume first that $d=0$, i.e., the
  reliability of task $T_i$ executed at speed $f_{i}$ is $R_i(f_{i}) =
  1-\lambda_0\frac{w_i}{f_{i}}$.  We show that executing task $T_i$
  twice at speed $f = \sqrt{f_i f'_i}$ improves the energy consumption
  while matching the deadline and reliability constraints.  Clearly
  the reliability constraint is matched, since $1-\lambda_0^2 w_i^2
  \frac{1}{f^2} = 1- \lambda_0^2 w_i^2 \frac{1}{f_i f'_i}$.  The fact
  that the deadline constraint is matched is due to the fact that
  $\sqrt{f_i f'_i} \geq \frac{2f_if'_i}{f_i+f'_i}$ (by squaring both
  sides of the equation we obtain $(f_i-f'_i)^2 \geq 0$). Then we use
  the fact that $f_d = \frac{2f_if'_i}{f_i+f'_i}$ is the minimal speed
  such that $\forall f \geq f_d,~ \frac{2w_i}{f} < \frac{w_i}{f_i} +
  \frac{w_i}{f'_i}$.  Finally, it is easy to see that the energy
  consumption is improved since $2 f_i f'_i \leq f^2_i + f'^2_i$,
  hence $2w_if_i f'_i \leq w_if^2_i + w_if'^2_i$.

  In the general case when $d\neq 0$, instead of having a closed form
  formula for the new speed $f$ common to both executions, we have
  $f=\max(f_1,f_2)$, where $f_1$~is dictated by the reliability
  constraint, while $f_2$ is dictated by the deadline constraint. 
  $f_1$ is the solution to the equation $2(dX + \ln X) = (df_i + \ln
  f_i) + (df'_i + \ln f'_i)$; this equation comes from the reliability
  constraint: the minimum speed~$X$ to match the reliability is
  obtained with
  $1-\lambda_0^2w_i^2\frac{e^{-df_i}}{f_i}\frac{e^{-df'_i}}{f'_i} =
  1-\lambda_0^2w_i^2\frac{e^{-2dX}}{X^2}$.  The deadline constraint
  must also be enforced, and hence $f_2=\frac{2f_if'_i}{f_i+f'_i}$
  (minimum speed to match the deadline).  Then the fact that the
  energy does not increase comes from the convexity of this function.

  Let $f$ be the unique speed at which the task is executed
  (twice). If $f\geq  \frac{1}{\sqrt2}\fr$, then executing the task
  only once at speed~\fr has  a lower energy consumption and execution
  time, while still matching the reliability constraint. Hence it is
  not optimal to re-execute the task unless $f<\frac{1}{\sqrt2}\fr$.  
  Finally, note that $f$ must be greater than \finf, solution 
  of Equation~(\ref{eq.fmin}), since \finf is the minimum
  speed such that the reliability constraint is met if task~$T_i$ is
  executed twice at the same speed. 
\end{proof}

Note that both lemmas can be applied to any solution of the
\tricritcont problem, not just optimal solutions, hence all heuristics
of Section~\ref{heurWC} will assign a unique speed to each task, be it
re-executed or not. 

We are now ready to assess the problem complexity: 
\begin{theorem}
\label{cont_all_exec}
The \chaincont problem is NP-hard, but not known to be in NP. 
\end{theorem}

Note that the problem is not known to be in NP because speeds could
take any real values (\continuous model). The completeness comes from
SUBSET-SUM~\cite{GareyJohnson}. The problem is NP-hard even for a
linear chain application mapped on a single processor (and any general
DAG mapped on a single processor becomes a linear chain). 

\begin{proof}
  Consider the associated decision problem: given a deadline, and
  energy and reliability bounds, can we schedule the graph to match
  all these bounds? Since the speeds could take any real values, the
  problem is not known to be in NP. For the
  completeness, we use a reduction from
  SUBSET-SUM~\cite{GareyJohnson}. Let $\II_1$ be an instance of
  SUBSET-SUM: given $n$ strictly positive integers $a_1, \ldots, a_n$,
  and a positive integer $X$, does there exist a subset $I$ of $\{1,
  \ldots, n\}$ such that $\sum_{i\in I}a_i= X $? Let $S=\sum_{i=1}^n
  a_i$.

 \smallskip
  We build the following instance~$\II_2$ of our problem. The
  execution graph is a linear chain with $n$~tasks, where:

\begin{compactitem}
\item task $T_i$ has weight $w_i=a_i$;
\item $\lambda_0 = \frac{\fmax}{100 \max_i a_i}$;
\item $\fmin = \sqrt{\lambda_0 \max_i a_i \fmax} = \frac{1}{10}\fmax$;
\item $\fr=\fmax$; $d=0$. 
\end{compactitem}

\medskip
\noindent The bounds on reliability, deadline and energy are: 
\begin{compactitem}
\item $R^0_i = R_i(\fr) = 1 - \lambda_0 \frac{w_i}{\fr}$ for $1\leq i
  \leq n$;
\item $D_0=  \frac{S}{\fr} + \frac{X}{c \fr}$, where
$c$ is the unique positive real root of the polynomial \mbox{$7y^{3}+21y^{2}-3y-1$.}
Analytically, we derive that 
$c = 4 \sqrt{\frac{2}{7}} \cos{\frac{1}{3}(\pi-\tan^{-1}{\frac{1}{\sqrt{7}}})}-1 ~ (\approx 0.2838)$;
but this value is irrational, so have to we encode it symbolically
rather than numerically; 
\item $E_0= 2X (\dfrac{2c}{1+c}\fr )^2 + (S - X) \fr^2$.
\end{compactitem}

\medskip
Clearly, the size of $\II_2$ is polynomial in the size of~$\II_1$. 

\medskip
\emph{Suppose first that instance $\II_1$ has a solution,~$I$.} For
all $i\in I$, $T_i$ is executed twice at
speed~$\dfrac{2c}{1+c}\fr$. Otherwise, for all $i \notin I$, it is
executed only once at speed~$\fr$.  The execution time is
$\sum_{i\notin I} \frac{a_i}{\fr} + \sum_{i \in I}
2\frac{a_i}{\frac{2c}{1+c}\fr} = \frac{S-X}{\fr} +
2X\frac{1+c}{2c\fr} = D_0$.  The reliability constraint is obviously
met for tasks not in~$I$. It is also met for all tasks in~$I$, since
$\dfrac{2c}{1+c}\fr > \fmin$, and two executions at \fmin are sufficient to
match the reliability constraint.  Indeed,
$1-\lambda_0^2\frac{a_i^2}{\fmin^2} = 1-\lambda_0\frac{a_i}{\fr}
 \frac{a_i}{\max_i a_i} \geq 1-\lambda_0\frac{a_i}{\fr} =
R^0_i$.  The energy consumption is exactly~$E_0$. All bounds are
respected, and therefore we have a solution to~$\II_2$.

\medskip
\emph{Suppose now that $\II_2$ has a solution.}  Let $I=\{i\; |\;
T_i\mbox{ is executed twice in the solution}\}$, and \mbox{$Y=\sum_{i\in
  I}a_i$.}  We prove in the following that necessarily $Y=X$, since the
energy constraint~$E_0$ is respected in~$\II_2$. 

We first point out that tasks executed only once are necessarily
executed at maximum speed to match the reliability constraint.  Then
consider the problem of minimizing the energy of a set of tasks, some
executed twice, some executed once at maximum speed, and assume that
we have a deadline~$D_0$ to match, but no constraint on reliability or
on \fmin. We will verify later that these additional two constraints
are indeed satisfied by the optimal solution when the only constraint
is the deadline.  Thanks to Lemma~\ref{WC_re-ex}, for all $i
\in I$,  task~$T_i$ is executed twice at the same speed. It is easy to
see that in fact all tasks in~$I$ are executed at the same speed,
otherwise we could decrease the energy consumption without modifying
the execution time, by convexity of the function. Let~$f$ be the speed
of execution (and re-execution) of task~$T_i$, with $i\in I$. 
Because the deadline is the only constraint, either $Y=0$
(no tasks are re-executed), or it is optimal to exactly match the
deadline~$D_0$ 
(otherwise we could just slow down all the re-executed tasks and
this would decrease the total energy).  Hence the problem amounts to
find the values of $Y$ and $f$ that minimize the function $E = 2Y f^2
+ (S - Y) \fr^2$, with the constraint $(S - Y)/\fr +2Y/f \leq D_0$.
First, note that if $Y=0$ then $E>E_0$, and hence $Y>0$ (since it
corresponds to a solution of~$\II_2$).   
Therefore, since the deadline is tight, we have $f= \frac{2Y}{D_0\fr-
  (S - Y)}\fr$, and finally the energy consumption can be expressed as 
\[
E(Y) = \left(\frac{(2Y)^3}{(D_0\fr- (S - Y))^2} + (S-Y) \right)\fr^2. 
\]

We aim at finding the minimum of this function. Let  
$\tilde{Y} = \frac{Y}{D_0\fr - S}$. Then we have 
$E(\tilde{Y}) = \left(\frac{(2\tilde{Y})^3}{(1+\tilde{Y})^2} + 
(\frac{S}{D_0\fr-S}-\tilde{Y}) \right) \times (D_0\fr-S)\fr^2$.
Differentiating,  we obtain
$$E'(\tilde{Y}) = 
\left (\frac{3\times 2^3\tilde{Y}^2}{(1+\tilde{Y})^2} - 
\frac{2^4 \tilde{Y}^3}{(1+\tilde{Y})^3} -1\right) (D_0\fr-S)\fr^2 \; .$$
Finally, $E'(\tilde{Y})=0$ if and only if 
\begin{equation}
	\label{def_c}
24\tilde{Y}^2(1+\tilde{Y}) - 16 \tilde{Y}^3 -(1+\tilde{Y})^3 = 0.
\end{equation}
The only positive solution of Equation~(\ref{def_c}) is $\tilde{Y}=c$,
and therefore the unique minimum of $E(Y)$ is obtained for 
$Y = c(D_0\fr-S) = X$.   

Note that for $Y=X$, we have $E=E_0$, and therefore any other
value of~$Y$ would not correspond to a solution.  
There remains to check that the 
solution matches both constraints on \fmin and on reliability, to
confirm the hypothesis on the speed of tasks that are
re-executed. Using the same argument as in the first part of the
proof, we see that the reliability constraint is respected when a task
is executed twice at~\fmin, and therefore we just need to check that
$f\geq \fmin$. For $Y=X$, we have
$f= \frac{2c}{1+c}\fr > \fmin$.  

 Altogether, we have $\sum_{i\in I}a_i = Y = X$, and therefore
$\II_1$~has a solution. This concludes the proof.
\end{proof}

Even if \chaincont is NP-hard, we can characterize an optimal solution
of the problem: 
\begin{prop}
\label{prop_WC_fr}
If $\fr < \fmax$, then in any optimal solution of \chaincont, 
either all tasks are executed only once, at constant speed
$\max(\frac{\sum_{i=1}^n w_i}{D}, \fr)$; or at least one task is
re-executed, and then all tasks that are not re-executed are
executed at speed~$\fr$.
\end{prop}

\begin{proof}
  Consider an optimal schedule. If all tasks are executed only once,
  the smallest energy consumption is obtained when using the constant
  speed $\frac{\sum_{i=1}^n w_i}{D}$. However if $\frac{\sum_{i=1}^n
    w_i}{D}<\fr$, then we have to execute all tasks at speed $\fr$ to
  match both reliability and deadline constraints.

  \medskip Now, assume that some task $T_i$ is re-executed, and assume
  by contradiction, that some other task $T_j$ is executed only once
  at speed $f_j>\fr$. Note that the common speed $f_{i}$ used in both
  executions of $T_i$ is smaller than $\fr$, otherwise we would not
  need to re-execute $T_{i}$.  We have $f_i < \fr < f_j$, and we prove
  that there exist values $f'_{i}$ (new speed of one execution
  of~$T_{i}$) and $f'_{j}$ (new speed of~$T_{j}$) such that $f_i <
  f'_i$,  $\fr \leq f'_j < f_j$, and the energy consumed with the new
  speeds is strictly smaller, while the execution time is unchanged.
  The constraint on reliability will also be met, since the speed of
  one execution of~$T_i$ is increased, while the speed of~$T_j$
  remains above the reliability threshold.  Note that we do not modify
  the speed of the re-execution of~$T_i$ (that remains~$f_i$), and the
  time and energy consumption of this execution are not accounted for
  in the equations. Also, we restrict to values such that $f'_i \leq f'_j$. 

\smallskip
Our problem writes: do there exist $\epsilon,~\epsilon'> 0$ such that
$$\begin{array}{l}
w_i f_i^2 + w_j f_j^2 \;>\; w_i (f_i+\epsilon')^2 + w_j
(f_j-\epsilon)^2;\\[.2cm] 
 D \;=\; \displaystyle\frac{w_i}{f_i} + \frac{w_j}{f_j} =
 \frac{w_i}{f_i+\epsilon'} + \frac{w_j}{f_j-\epsilon};\\[.3cm]
f_i \;<\; f_i+\epsilon' \;\leq\; f_j-\epsilon;\\
\fr \;\leq\; f_j-\epsilon \;<\; f_j .
\end{array}$$ 

We study the function $\phi : \epsilon \mapsto  w_i f_i^2 + w_j f_j^2 -
\left( w_i (f_i+\epsilon')^2 + w_j (f_j-\epsilon)^2 \right)$, and we
want to prove that it is positive. Thanks to the deadline constraint
($D$~is the bound on the execution time of~$T_j$ plus one execution
of~$T_i$), 
we have $f_i = \frac{w_i f_j}{Df_j - w_j}$, and 
$f_i + \epsilon' = \frac{w_i}{D-\frac{w_j}{f_j - \epsilon}} 
   = \frac{w_i(f_j - \epsilon)}{D (f_j - \epsilon)- w_j}$. 

We can therefore express $\phi(\epsilon)$ as: 
$$\phi(\epsilon)= \frac{w_i^3 f_j^2}{(D f_j - w_j)^2} -
\frac{w_i^3 (f_j- \epsilon)^2}{(D(f_j-\epsilon) -w_j)^2} +
w_j f_j^2 - w_j (f_j-\epsilon)^2. $$

Moreover, we study the function for $\epsilon>0$, and because of the
constraint on new speeds, $\epsilon \leq f_j - \fr$. Another bound on
$\epsilon$ is obtained from the fact that $f_i+\epsilon' \leq f_j
-\epsilon$, and the equality is obtained when both tasks are running
at speed~$ \frac{w_i+w_j}{D}$, thus meeting the deadline. Hence, $f_j
-\epsilon \geq \frac{w_i+w_j}{D}$, and finally
$$0< \epsilon \leq f_j - \max\left(\fr,
  \dfrac{w_i+w_j}{D}\right). $$

Differentiating, we obtain $$\phi'(\epsilon) = 
\frac{2w_i^3 (f_j- \epsilon)}{(D(f_j\!-\!\epsilon) \!-\! w_j)^2} -
\frac{2Dw_i^3 (f_j- \epsilon)^2}{(D(f_j\!-\!\epsilon) \!-\! w_j)^3} +2w_j
(f_j-\epsilon). $$
We are looking for $\epsilon$ such that $\phi'(\epsilon)=0$, hence
obtaining the polynomial 
$$X^3 - \frac{w_i^3}{w_j^3} = 0,$$ 
by multiplying each side of the equation by $\frac{(D(f_j-\epsilon) -
  w_j)^3}{w_j^4(f_j- \epsilon)}$, and defining
$X=\frac{D(f_j-\epsilon)-w_j}{w_j}$.  
The only real solution to this polynomial is $X = \frac{w_i}{w_j}$,
that corresponds to $\epsilon = f_j - \frac{w_i+w_j}{D}$. 
Therefore, the only extremum of the function~$\phi$ is obtained for
this value of~$\epsilon$, which corresponds to executing both tasks at
the same speed. Because of the convexity of the energy consumption,
this value corresponds to a maximum of function~$\phi$ (see for
instance Proposition~2 in~\cite{aupy12ccpe}), since the energy is
minimized when both tasks run at the same speed. Therefore, $\phi$~is
strictly increasing for $0 \leq \epsilon \leq f_j - \frac{w_i+w_j}{D}$,
and for $\epsilon = f_j - \max\left(\fr,\dfrac{w_i+w_j}{D}\right)$,
$\phi$~is maximal (with regards to our constraints), and
$\phi(\epsilon)>0$.  

\medskip

Altogether, this value of $\epsilon$ gives us two new speeds $f'_i =
\frac{w_i(f_j - \epsilon)}{D (f_j - \epsilon)- w_j}$ 
and $f'_j = f_j - \epsilon$ that strictly improve the energy consumption of the
schedule, while the constraints on deadline and reliability are still
enforced. However, the original schedule was supposed to be optimal,
we have a contradiction, which concludes the proof. 
\end{proof}

In essence, Proposition~\ref{prop_WC_fr} states that when dealing with
a linear chain, we should first slow down the execution of each task
as much as possible. Then, if the deadline is not too tight, i.e., if
$\fr > \frac{\sum_{i=1}^n w_i}{D}$, there remains the possibility to
re-execute some of the tasks (and of course it is NP-hard to decide
which ones). Still, this general principle \emph{``first slow-down
and then re-execute''} will guide the design of type A heuristics in
Section~\ref{heurWC}.

While the general \tricritcont problem is NP-hard even with a single
processor, the particular variant \forkcont can be solved in
polynomial time: 
\begin{theorem}
\label{thm_fork}
The \forkcont problem can be solved in polynomial time. 
\end{theorem}

The difficulty to provide an optimal algorithm for the \forkcont
problem comes from the fact that the total execution
time must be shared between the source of the fork,~$T_0$, and the other tasks that
all run in parallel. If we know~$D'$, the fraction of the deadline
allotted for tasks $T_1,\dots,T_n$ once the source has finished its
execution, then we can decide which tasks are re-executed and all
execution speeds. 

\begin{proof} 
  We start by showing that
  \tricritcont can be solved in polynomial time for one single task,
  and then for $n$ independent tasks, before tackling the problem
  \forkcont.  

\paragraph{\tricritcont for a single task on one processor 
can be solved in polynomial time.}

When there is a single task $T$ of weight $w$, the solution depends on
the deadline~$D$:  
\begin{enumerate}
    \item if $D<\frac{w}{\fmax} = D^{(0)}$, then there is no solution;
    \item if $\frac{w}{\fmax} \leq D\leq \frac{w}{\fr}= D^{(1)}$, then
      $T$ is executed once at speed $\frac{w}{D}$, the minimum energy
      is $w^3\times\frac{1}{D^2}$;
    \item if $\frac{w}{\fr} < D \leq \frac{2\sqrt{2}w}{\fr}=
      D^{(2)}$, then $T$ is executed once at speed $\fr$, the minimum
      energy is $w\fr^2$;
    \item if $\frac{2\sqrt{2}w}{\fr} < D \leq
      \frac{2w}{\finfc}= D^{(3)}$, then $T$ is executed
      twice at speed $\frac{2w}{D}$, the minimum energy is
      $(2w)^3\times\frac{1}{D^2}$;
    \item if $\frac{2w}{\finfc} < D$, then $T$ is
      executed twice at speed $\finfc$, the minimum energy
      is $2wf^{(\texttt{inf})2}$.
\end{enumerate}
These results are a direct consequence from the deadline and
reliability constraints. With a deadline smaller than~$D^{(0)}$, the
task cannot be executed within the deadline, even at speed~\fmax. 
The bound~$D^{(2)}$ comes from Lemma~\ref{WC_re-ex}, 
which states that we need to have enough time to execute the task
twice at a speed lower than $\frac{1}{\sqrt2}\fr$ before re-executing
it. Therefore, the task is executed only once for smaller deadlines,
either at speed~$w/D$, or at speed~\fr if $w/D<\fr$. For larger
deadlines, the task is re-executed, either at speed~$2w/D$, or at
speed~\finfc if $2w/D<\finfc$, since the re-execution speed cannot be
lower than~\finfc (see Lemma~\ref{WC_re-ex}). 

\paragraph{\tricritcont for $n$ independent tasks on $n$ processors
can be solved in polynomial time.}

For $n$ independent tasks mapped on $n$ distinct processors, decisions
for each task can be made independently, and we simply solve $n$~times
the previous single task problem. The minimum energy is the sum
of the minimum energies obtained for each task.

\paragraph{\forkcont.}
For a fork, we need to decide how to share the deadline between the
source~$T_0$ of the fork and the other tasks (i.e., $n$ independent
tasks on $n$~processors). We search the optimal values $D_1$ and $D_2$
such that $D_1+D_2=D$, and the energy of executing $T_0$ within
deadline~$D_1$ plus the energy of executing all other tasks
within~$D_2$ is minimum. Therefore, we just need to find the optimal
value for~$D_2$ (since $D_1=D-D_2$), and reuse previous results for
independent tasks. 

Independently of $D$, we can define for each task~$T_i$ four values
$D_i^{(0)}, D_i^{(1)}, D_i^{(2)}$ and $D_i^{(3)}$, as in the case of a
single task. There is a solution if and only if $\max_{1\leq i \leq n}
D_i^{(0)}\leq D_2\leq D-D_0^{(0)}$. Then, the energy consumption
depends upon the intervals delimited by values $D-D_0^{(j)}$ and
$D_i^{(j)}$, for $1\leq i \leq n$ and $j=1,2,3$.
Within an interval, the energy consumed by the source is either a
constant, or a constant times $\frac{1}{(D-D_2)^2}$, and the energy
consumed by task~$T_i$ ($1\leq i \leq n$) is either a constant, or 
a constant times~$\frac{1}{D_2^2}$.  
All the constants are known, only dependent of~$T_i$, and they are
obtained by the algorithm that gives the optimal solution to
\tricritcont for a single task.  
To obtain the intervals, we sort the $4n$ values of~$D_i^{(j)}$
($i>0$) and the four values of $D-D_0^{(j)}$, with $j=0,1,2,3$, and
rename these $4(n+1)$ values as~$d_k$, with $1\leq k \leq 4(n+1)$ and
$d_k \leq d_{k+1}$.   
Given the bounds on~$D_2$, we consider the intervals of
the form $[d_k,d_{k+1}]$, with $d_k \geq \max_{1\leq i \leq n}
D_i^{(0)}$, and $d_{k+1} \leq D-D_0^{(0)}$. 
On each of these intervals, 
the energy function is $\frac{K}{(D-D_2)^2} + \frac{K'}{D_2^2} + K''$,
where $K$, $K'$ and $K''$ are positive constants that can be obtained
in polynomial time by the solution to \tricritcont for a single task.
Finding a minimum to this function on the interval $[d_k,d_{k+1}]$ 
can be done in polynomial time: 
\begin{itemize}
\item the first derivative of this function is $\frac{2K}{(D-D_2)^3} -
  \frac{2K'}{D_2^3}$;
\item the function is convex on $]0,D[$, indeed the second derivative
  of this function is $\frac{6K}{(D-D_2)^4} + \frac{6K'}{D_2^4}$,
  which is positive on $]0,D[$, and therefore on the interval
  $[d_k,d_{k+1}]$, there is exactly one minimum to the energy function
  ($d_k>0$ and $d_{k+1}<D$); 
\item the minimum is obtained either when the first derivative is
  equal to zero in the interval (i.e., if there is a solution to the
  equation $2KD_2^3 - 2K'(D-D_2)^3 = 0$ in $[d_k,d_{k+1}]$), or the
  minimum is reached at $d_k$ (resp. $d_{k+1}$) if the first
  derivative is positive (resp. negative) on the interval.  
\end{itemize}

There are $O(n)$ intervals, and it takes constant time to find the
minimum energy~$E_k$ within interval~$[d_k,d_{k+1}]$, as explained
above, by solving one equation. Since we have partitioned the interval
of possible deadlines $D_2 \in \left[\max_{1\leq i \leq n} D_i^{(0)},
  D-D_0^{(0)}\right]$, and obtained the minimum energy consumption in
each sub-interval, the minimum energy consumption for the fork graph
is $\min_k E_k$, and the value of~$D_2$ is obtained where the minimum
is reached. Once we know the optimal value of~$D_2$, it is easy to
reconstruct the solution, following the algorithm for a single task,
in polynomial time. 
\end{proof}

Note that this algorithm does not provide any closed-form formula for the
speeds of the tasks, and that there is an intricate case analysis due
to the reliability constraints.

If we further assume that the fork is made of identical tasks (i.e.,
$w_i=w$ for $0\leq i \leq n$), then we can provide a closed-form
formula. However, Proposition~\ref{prop_WC_fork} illustrates the
inherent difficulty of this {\em simple} problem, with several cases
to consider depending on the values of the deadline, and also the
bounds on speeds ($\fmin$, $\fmax$, $\fr$, etc.). First, since the
tasks all have the same weight~$w_i=w$, we get rid of the $\finf$
introduced above, since they are all identical (see
Equation~(\ref{eq.fmin})): $\finf=\finfc$ for $0\leq i \leq n$. 
Therefore we let $\fmin = \max(\fmin,\finfc)$ in the proposition
below:

\begin{prop}
\label{prop_WC_fork}
In the optimal solution of \forkcont with at least three identical
tasks (and hence $n\geq 2$), 
there are only three possible scenarios: (i) no task is re-executed; 
(ii) the $n$ successors are all re-executed
but not the source; (iii) all tasks are re-executed.
In each scenario, the source is executed at speed $\sss$ (once or
twice), and the $n$ successors are
executed at the same speed $\ttt$ (once or twice). \\
For a deadline $D<\frac{2w}{\fmax}$, there
is no solution. For a deadline $D \in \left[\frac{2w}{\fmax}, \frac{w}{\fr}
\frac{(1+2n^{\frac{1}{3}})^{\frac{3}{2}}}{\sqrt{1+n}}\right]$, no task
is re-executed (scenario (i)) and the values of $\sss$ and $\ttt$ are
the following:
\begin{compactitem} 
\item if $\frac{2w}{\fmax}\leq D\leq
  \min\left(\frac{w}{\fmax}(1+n^{\frac{1}{3}}),w(\frac{1}{\fr}+\frac{1}{\fmax})\right)$,
  then $\sss = \fmax$ and $\ttt= 
  \frac{w}{D\fmax-w}\fmax$;
\item if $\frac{w}{\fmax}(1+n^{\frac{1}{3}}) \leq w(\frac{1}{\fr}+\frac{1}{\fmax})$, then
\begin{compactitem}
\item if $\frac{w}{\fmax}(1+n^{\frac{1}{3}}) < D\leq \frac{w}{\fr}
  \frac{1+n^{\frac{1}{3}}}{n^{\frac{1}{3}}}$, then $\sss =
  \frac{w}{D}(1+n^{\frac{1}{3}})$ and $\ttt =
  \frac{w}{D}\frac{1+n^{\frac{1}{3}}}{n^{\frac{1}{3}}}$;
\item if $\frac{w}{\fr} \frac{1+n^{\frac{1}{3}}}{n^{\frac{1}{3}}}
  <D\leq \frac{2w}{\fr}$, then $\sss = \frac{w}{D\fr-w}\fr$ and $\ttt
  = \fr$;
\end{compactitem}
\item if $\frac{w}{\fmax}(1+n^{\frac{1}{3}}) > w(\frac{1}{\fr}+\frac{1}{\fmax})$, then
\begin{compactitem}
\item if $w(\frac{1}{\fr}+\frac{1}{\fmax}))
  <D\leq \frac{2w}{\fr}$, then $\sss = \frac{w}{D\fr-w}\fr$ and $\ttt
  = \fr$;
\end{compactitem}
\item if $\frac{2w}{\fr} <D\leq \frac{w}{\fr}
  \frac{(1+2n^{\frac{1}{3}})^{\frac{3}{2}}}{\sqrt{1+n}}$, then $\sss =
  \ttt = \fr$.
\end{compactitem}
\end{prop}

Note that for larger values of $D$, depending on \fmin, we can move to
scenarios (ii) and (iii) with partial or total re-execution.  The case
analysis becomes even more painful, but remains feasible. Intuitively,
the property that all tasks have the same weight is the key to
obtaining analytical formulas, because all tasks have the same minimum
speed~$\finfc$ dictated by Equation~(\ref{eq.fmin}).  

\begin{proof}
First, we recall preliminary results:
\begin{compactitem}
\item if a task is executed only once at speed~$f$, then $\fr\leq f \leq \fmax$;
\item if a task is re-executed, then both executions are done at the
  same speed~$f$, and \mbox{$\fmin \leq f < \frac{1}{\sqrt{2}}\fr$.}
\end{compactitem}

By hypothesis, all tasks are identical: the bound on re-execution
speed accounts for \finfc as in Lemma~\ref{WC_re-ex}, since we now
have $\fmin=\max(\fmin,\finfc)$.
Therefore, if two tasks of same weight~$w$ have the same energy
consumption in the optimal solution, then they are executed the same
number of times (once or twice) and at the same speed(s).
If the energy is greater than or
equal to $w\fr^2$, then necessarily there is one execution; and if it is
lower than $w\fr^2$, then necessarily there are two executions.

First, we prove that in any solution, the energy consumed for the
execution of each successor task, also called {\em leaf}, is the same.
If it was not the case, since each task has the same weight, and since
each leaf is independent from the other and only dependent on the source of
the fork, if a leaf~$T_i$ is consuming more than another leaf~$T_j$,
then we could execute~$T_i$ the same number of times and at the same
speed than~$T_j$, hence matching the
deadline bound and the reliability constraint, and obtaining a better
solution.  Thanks to this result, we now assume that all leaves are
executed at the same speed(s), denoted~$\ttt$. The source task may be
executed at a different speed,~$\sss$.

Next, let us show that the energy consumption of the source is always
greater than or equal to that of any leaf in any optimal solution. 
First, since the source and leaves have the same weight, if we invert
the execution speeds of the source and of the leaves, then the
reliability of each task is still matched, and so is the execution
time. Moreover, the energy consumption is equal to the energy
consumption of the source plus $n$ times the energy consumption of any
leaf (recall that they all consume the same amount of energy). Hence,
if the energy consumption of the source is smaller than the one of the
leaves, permuting those execution speeds would reduce by $(n-1) \times
\Delta$ the energy, where $\Delta$ is the positive difference between
the two energy consumptions. Thanks to this result, we can say that
the source should never be executed twice if the leaves are executed
only once since it would mean a lower energy consumption for the
source (recall that $n\geq 2$). 

\medskip 

This result fully characterizes the shape of any optimal
solution. There are only three possible scenarios: (i) no task is
re-executed; (ii) the $n$ successors (leaves) are all re-executed but
not the source; (iii) all tasks are re-executed. 
We study independently the three scenarios, i.e., we aim at
determining the values of \sss and \ttt in each case. 
Conditions on the deadline indicate the shape of the solution, and we
perform the case analysis for deadlines $D\leq \frac{w}{\fr}
\frac{(1+2n^{\frac{1}{3}})^{\frac{3}{2}}}{\sqrt{1+n}}$.

\medskip
Let us assume first that the optimal solution is such that each task
is executed only once (scenario~(i)).  
From the proof of Theorem~1 in~\cite{aupy12ccpe}, 
we obtain the optimal speeds with no re-execution and without
accounting for reliability; they are given by 
the following formulas:

\begin{compactitem}
\item if $D<\frac{2w}{\fmax}$, then there is no solution, since the
  tasks executed at~\fmax exceed the deadline;
\item if $\frac{2w}{\fmax}\leq D \leq \frac{w}{\fmax}(1+n^{\frac{1}{3}})$, 
  then $\sss = \fmax$ and $\ttt = \frac{w}{D\fmax-w}\fmax$;
\item if $\frac{w}{\fmax}(1+n^{\frac{1}{3}}) < D $, then 
  $\sss = \frac{w}{D}(1+n^{\frac{1}{3}})$ and 
  $\ttt = \frac{w}{D}\frac{1+n^{\frac{1}{3}}}{n^{\frac{1}{3}}}$.
\end{compactitem}

\medskip 
Since there is a minimum speed~\fr to match the
reliability constraint, there is a condition when $\ttt<\fr$ that makes an
amendment on some of the items. Note that in all cases, if
$D>\frac{2w}{\fr}$, then both the source and the leaves are executed
at speed~\fr, i.e., $\sss=\ttt=\fr$ (recall that we consider the case
with no re-execution). 
\begin{compactitem}
\item If $\frac{2w}{\fmax}\leq D \leq
  \frac{w}{\fmax}(1+n^{\frac{1}{3}})$, then we need $\ttt=
  \frac{w}{D\fmax-w}\fmax \geq \fr$, hence the condition: $D\leq
  \min\left(\frac{w}{\fmax}(1+n^{\frac{1}{3}}),w(\frac{1}{\fr}+
    \frac{1}{\fmax})\right)$. In 
  this case, $\sss = \fmax$ and $\ttt=\frac{w}{D\fmax-w}\fmax$.
\item If $D > \min\left(\frac{w}{\fmax}(1+n^{\frac{1}{3}}),
          w(\frac{1}{\fr}+\frac{1}{\fmax})\right)$, then the previous
          results do not hold anymore because of the constraint on the
          speed of the leaves. We must further differentiate cases,
          depending on where the minimum is reached. 
\item If $\frac{w}{\fmax}(1+n^{\frac{1}{3}}) \leq
  w(\frac{1}{\fr}+\frac{1}{\fmax})$, then 
  \begin{compactitem}
        \item 
           if  $\frac{w}{\fmax}(1+n^{\frac{1}{3}}) < D\leq \frac{w}{\fr}
            \frac{1+n^{\frac{1}{3}}}{n^{\frac{1}{3}}}$, 
we are in the third case with no reliability, and therefore 
  $\sss = \frac{w}{D}(1+n^{\frac{1}{3}})$ and 
  $\ttt = \frac{w}{D}\frac{1+n^{\frac{1}{3}}}{n^{\frac{1}{3}}}$; 
  the upper bound on~$D$ guarantees that \mbox{$\ttt \geq \fr$}, while
  the lower bound on~$D$ guarantees that  \mbox{$\sss \leq \fmax$}; 
	\item if $\frac{w}{\fr} \frac{1+n^{\frac{1}{3}}}{n^{\frac{1}{3}}} <D\leq  \frac{2w}{\fr}$, 
	then the speed of the leaves is constrained by \fr, and we
        obtain $\ttt = \fr$ and 
$\sss  = \frac{w}{D\fr-w}\fr$. From the lower bound on~$D$, we obtain
$\sss<n^{\frac{1}{3}} \fr$, and since $\frac{w}{\fmax}(1+n^{\frac{1}{3}}) \leq
  w(\frac{1}{\fr}+\frac{1}{\fmax})$, we have $\sss<n^{\frac{1}{3}} \fr \leq
  \fmax$. 
  \end{compactitem}
\item If $\frac{w}{\fmax}(1+n^{\frac{1}{3}}) >
  w(\frac{1}{\fr}+\frac{1}{\fmax})$, then 
for $ w(\frac{1}{\fr}+\frac{1}{\fmax})<D \leq   \frac{2w}{\fr}$, 
the leaves should be
  executed at speed~$\ttt=\fr$, and for the source,
  $\sss=\frac{w}{D\fr-w}\fr$.  Note that the lower bound on~$D$ is
  equivalent to $\frac{w}{D\fr-w}\fr < \fmax$, and hence the speed of
  the source is not exceeding~\fmax. 
\end{compactitem}

\medskip
As stated above, if $D>\frac{2w}{\fr}$, both the source and the leaves
are executed at speed~\fr (with no re-execution). However, if the
deadline is larger, re-execution will be used by the optimal solution
(i.e., it will become scenario~(ii)). Let us consider therefore the
scenario in which leaves are re-executed, to compare the energy
consumption with the first scenario. 
In this case, we consider an equivalent fork in which leaves are of
weight~$2w$, and a schedule with no re-execution.  
Then the optimal solution when there is no maximum speed is: 
$$\sss = \frac{w}{D}(1+2n^{\frac{1}{3}}) \quad \text{and} \quad 
\ttt = \frac{w}{D}\frac{1+2n^{\frac{1}{3}}}{n^{\frac{1}{3}}}\; .$$ 

If $\ttt \geq \frac{1}{\sqrt{2}}\fr$, then there is a better solution
to the original problem without re-execution. 
Indeed, the solution in which the leaves (of weight~$w$) 
are executed once at speed
\mbox{$\ttt'=\max(\ttt,\fr)$} 
is such that:
 \begin{compactitem}
\item the reliability constraint is matched ($\ttt'\geq\fr$); 
\item the deadline constraint is matched ($\ttt'\geq \ttt$, and \ttt
  corresponds to the solution with re-execution, i.e., $w/\sss +
  2w/\ttt \leq D$); 
\item the energy consumption is better, as stated by
  Lemma~\ref{WC_re-ex} if $\ttt'=\fr$.  
\end{compactitem}

\noindent 
Therefore, we are in scenario~(ii) when $\ttt <
\frac{1}{\sqrt{2}}\fr$, i.e., 
$D>\frac{w}{\fr} \sqrt{2}\frac{1+2n^{\frac{1}{3}}}{n^{\frac{1}{3}}}$. \\
Moreover,  
depending whether $\sss \geq \fr$ or $\sss<\fr$:
\begin{compactitem}
\item if $\sss \geq \fr$, i.e.,  $D \leq \frac{w}{\fr}(1+2n^{\frac{1}{3}})$,
then the solution is valid; 
\item if $\sss < \fr$, 
then we must in fact have $\sss = \fr$, and then \mbox{$\ttt =
  \max(\frac{2w}{D\fr - w}\fr, \fmin)$.}
\end{compactitem}

\noindent Note that these values do not take into account the
constraints \fmax and \fmin. Therefore, they are lower bounds on the
energy consumption when the leaves are re-executed.

\medskip Finally, we establish a bound~$D_0$ on the deadline: for
larger values than~$D_0$, we cannot guarantee that re-execution will
not be used by the optimal solution, and hence we will have fully
characterized the cases for deadlines smaller than~$D_0$. Since we
have only computed lower bounds on energy consumption for the
scenario~(ii), this bound will not be tight.
We know that the minimum energy consumption is a function decreasing
with the deadline: if $D > D'$, then any solution for $D'$ is a
solution for~$D$. Let us find the minimum deadline~$D$ such that the
energy when the leaves are re-executed is smaller than the energy when
no task is re-executed.

As we have seen before, necessarily if $D\leq \frac{w}{\fr}
\sqrt{2}\frac{1+2n^{\frac{1}{3}}}{n^{\frac{1}{3}}}$, then it is better
to have no re-execution, i.e., $D_0 \geq  \frac{w}{\fr}
\sqrt{2}\frac{1+2n^{\frac{1}{3}}}{n^{\frac{1}{3}}}$.  
Let $D = \frac{w}{\fr}
\sqrt{2}\frac{1+2n^{\frac{1}{3}}}{n^{\frac{1}{3}}} + \epsilon$.
We suppose also that \mbox{$D \leq \frac{w}{\fr}(1+2n^{\frac{1}{3}})$}, i.e.,
the solution with re-execution is valid ($\sss\geq \fr$). 
\begin{compactitem}
\item The energy consumption when the leaves are re-executed is
  greater than\\ $E_2 = w \sss^2 + 2nw \ttt^2 = 
    \frac{w^3}{D^2}(1+2n^{\frac{1}{3}})^3$. 
\item With no re-execution, the deadline is large enough so that each
  task can be executed at speed~\fr, and therefore the energy
  consumption is\\ $E_1 = (1+n)w\fr^2 =
  2\frac{w^3}{(D-\epsilon)^2}(1+n)\left
    (\frac{1+2n^{\frac{1}{3}}}{n^{\frac{1}{3}}}\right )^2$.
\end{compactitem}

\noindent We now check the condition $E_1 \leq E_2$: 
\begin{align*}
2\frac{w^3}{(D-\epsilon)^2}(1+n)\left (\frac{1+2n^{\frac{1}{3}}}{n^{\frac{1}{3}}}\right )^2 &\leq \frac{w^3}{D^2}(1+2n^{\frac{1}{3}})^3\\
\frac{2}{(D-\epsilon)^2}\frac{1+n}{n^{\frac{2}{3}}} &\leq \frac{1+2n^{\frac{1}{3}}}{D^2} \\
\frac{D^2}{(D-\epsilon)^2} &\leq \frac{n^{\frac{2}{3}}+2n}{2+2n} \\
D &\leq \frac{w}{\fr}
\frac{(1+2n^{\frac{1}{3}})^{\frac{3}{2}}}{\sqrt{1+n}} = D_0
\end{align*}

Furthermore, note that $D_0 < \frac{w}{\fr}(1+2n^{\frac{1}{3}})$ for
$n>2$, hence the hypothesis that $\sss\geq \fr$ is valid for the
values considered. 
Finally, if the deadline is smaller than the threshold value~$D_0$, 
then we can guarantee that the optimal solution will not do any
re-execution. However, if the deadline is larger, we do not know what
happens (but it can be computed as a function of \fmin, \fmax and \fr).
\end{proof}

Beyond the case
analysis itself, the result of Proposition~\ref{prop_WC_fork} is
interesting: we observe that in all cases, the source task is executed
faster than the other tasks. This shows that
Proposition~\ref{prop_WC_fr} does not hold for general DAGs, and
suggests that some tasks may be more critical than others. A
hierarchical approach, that categorizes tasks with different
priorities, will guide the design of type~B heuristics in
Section~\ref{heurWC}.

\section{\VDD model}
\label{vdd_model}

Contrarily to the \continuous model, the \VDD model uses discrete
speeds. A processor can choose among a set $\{f_1,...,f_m\}$ of
possible speeds. A task can be executed at different speeds.  

Let $\alpha_{(i,j)}$ be the time of computation of task~$T_i$ at
speed~$f_j$. The execution time of a task $T_i$ is $\exe(T_i) =
\sum_{j=1}^m \alpha_{(i,j)}$, and the energy consumed during the
execution is $E_i = \sum_{j=1}^m \alpha_{(i,j)}f_j^3$.  Finally,
for the reliability, the approximation used in
Equation~(\ref{rel_first_order}) still holds. However, the reliability
of a task is now the product of the reliabilities for each time
interval with constant speed, hence
 $R_i = \prod_{j=1}^m (1-\lambda_0
\; e^{-df_j} \alpha_{(i,j)})$. Using
a first order approximation, we obtain\\[-.2cm]
\begin{equation}
	\label{vdd_rel_first_order}
R_i = 1-\lambda_0 \sum_{j=1}^m e^{-df_j} \alpha_{(i,j)} = 1-\lambda_0
\sum_{j=1}^m h_{j} \alpha_{(i,j)}, \text{ where  } h_j = e^{-df_j}, 1 \leq j \leq m.
\end{equation}

We first show that only two different speeds are needed for the
execution of a task.  
This result was already known for the bi-criteria problem
makespan/energy, and it is interesting to see that reliability does
not alter it:

\begin{prop}
    \label{vdd.2.speeds}
    With the \VDD model, each task is computed using at most two
    different speeds.
\end{prop}

\begin{proof}
  Suppose that a task is computed with three speeds, $f_1\leq f_2 \leq
  f_3$, and let $h_j= e^{-df_j}$, for $j=1,2,3$.  We show that we can
  get rid of one of those speeds. The proof will follow by induction.
  Let $\alpha_i$ be the time spent by the processor at speed $f_i$.
  We aim at replacing each $\alpha_i$ by some $\alpha'_i$ so that we
  have a better solution. The constraints write:
\begin{compactenum}
	\item Deadline not exceeded: \\[-.4cm]
\begin{equation} \label{D}\alpha_1 + \alpha_2 + \alpha_3
\geq \alpha'_1 + \alpha'_2 + \alpha'_3. \end{equation}

	\item Same amount of work:\\[-.4cm]
\begin{equation} \label{W}\alpha_1f_1 + \alpha_2f_2 + \alpha_3f_3
= \alpha'_1f_1 + \alpha'_2f_2 + \alpha'_3f_3. \end{equation}

	\item Reliability preserved:\\[-.4cm]
\begin{equation} \label{R}\alpha_1h_1 + \alpha_2h_2 + \alpha_3h_3
\geq \alpha'_1h_1 + \alpha'_2h_2 + \alpha'_3h_3. \end{equation}

	\item Less energy spent:\\[-.4cm]
\begin{equation} \label{E}\alpha_1f_1^3 + \alpha_2f_2^3 + \alpha_3f_3^3
> \alpha'_1f_1^3 + \alpha'_2f_2^3 + \alpha'_3f_3^3. \end{equation}
\end{compactenum}
\medskip

\noindent
We show that 
    $\alpha_1'=\alpha_1-\epsilon_1$, 
    $\alpha_2'=\alpha_2+\epsilon_1+\epsilon_3$, and
    $\alpha_3'=\alpha_3-\epsilon_3$ is a valid solution:
\begin{compactitem}
	\item Equation (\ref{D}) is satisfied, since 
$\alpha_1 + \alpha_2 + \alpha_3 = \alpha'_1 + \alpha'_2 + \alpha'_3$.
	\item Equation (\ref{W}) gives
$\epsilon_1 = \epsilon_3 \left ( \dfrac{f_3-f_2}{f_2-f_1} \right )$.
	\item Next we replace the $\alpha'_i$ and $\epsilon_i$ in Equation 
(\ref{R}) and we obtain
$h_2(f_3-f_1) \leq h_1(f_3-f_2) + h_3(f_2-f_1)$,
which is always true by convexity of the exponential 
(since $h_j=e^{-df_j}$).
	\item Finally, Equation (\ref{E}) gives us
$\epsilon_1f_1^3 + \epsilon_3f_3^3 > (\epsilon_3+\epsilon_1)f_2^3$, 
which is necessarily true since $f_1<f_2<f_3$ and $f\rightarrow f^3$ is
convex (barycenter).
\end{compactitem}
Since we  want all the $\alpha'_i$ to be nonnegative, we take 
$$
\epsilon_1 =
\min \left(\alpha_1, \alpha_3\left ( \dfrac{f_3-f_2}{f_2-f_1} \right )
\right) 
\quad \text{and} \quad 
\epsilon_3 =
\min \left(\alpha_3, \alpha_1\left ( \dfrac{f_2-f_1}{f_3-f_2} \right ) \right).
$$

We have either $\epsilon_1 = \alpha_1$ or $\epsilon_3 = \alpha_3$,
which means that $\alpha'_1=0$ or $\alpha'_3=0$, and we 
can indeed compute the task with only two speeds, meeting
the constraints and with a smaller energy. 
\end{proof}

We are now ready to assess the problem complexity:
\begin{theorem}
\label{vddk_all_exec}
The \chainvdd problem is NP-complete.
\end{theorem}

The proof is similar to that of Theorem~\ref{cont_all_exec}, assuming
that there are only two available speeds, \fmin and~\fmax. Then we
reduce the problem from SUBSET-SUM.  Note that here again, the problem
turns out to be NP-hard even with one single processor (linear chain
of tasks). 

\begin{proof}
  Consider the associated decision problem: given an execution graph,
  $m$ possible speeds, a deadline, a reliability, and a bound on the
  energy consumption, can we find the time each task will spend at
  each speed such that the deadline, the reliability and the bound on
  energy are respected?  The problem is clearly in NP: given the time
  spent in each speed for each task, computing the execution time, the
  reliability and the energy consumption can be done in polynomial
  time.  To establish the completeness, we use a reduction from
  SUBSET-SUM~\cite{GareyJohnson}. Let $\II_1$ be an instance of
  SUBSET-SUM: given $n$ strictly positive integers $a_1, \ldots, a_n$,
  and a positive integer $X$, does there exist a subset $I$ of $\{1,
  \ldots, n\}$ such that $\sum_{i\in I}a_i= X $? Let $S=\sum_{i=1}^n
  a_i$.

\medskip
We build the following
instance~$\II_2$ of our problem. The execution graph is a linear chain
with $n$~tasks, where:
\begin{compactitem}
\item task $T_i$ has weight $w_i=a_i$;
\item the processor can run at $m=2$ different speeds, $\fmin$ and $\fmax$;
\item $\lambda_0 = \dfrac{\fmax}{100 \max_{i=1..n} a_i}$;
\item $\fmin = \sqrt{\lambda_0 \fmax \max_{i=1..n} a_i} = \frac{\fmax}{10}$;
\item $\fr=\fmax$; $d=0$. 
\end{compactitem}

\medskip
\noindent The bounds on reliability, deadline and energy are: 
\begin{compactitem}
\item $R^0_i = R_i(\fr) = 1 - \lambda_0
  \frac{w_i}{\fr}$ for $1\leq i \leq n$; 
\item $D_0=\frac{2X}{\fmin} + \frac{S-X}{\fmax} $;
\item $E_0= 2X \fmin^2 + (S-X)\fmax^2$.
\end{compactitem}

\medskip
Clearly, the size of $\II_2$ is polynomial in the size of $\II_1$.

\medskip \emph{Suppose first that instance $\II_1$ has a
  solution,~$I$.} For all $i\in I$, $T_i$ is executed twice at
speed~$\fmin$. Otherwise, for all $i \notin I$, it is executed at
speed~$\fmax$ one time only.  The execution time is $\frac{2\sum_{i\in
    I}a_i}{\fmin} + \frac{\sum_{i\notin I}a_i}{\fmax} =
\frac{2X}{\fmin} + \frac{S-X}{\fmax} = D$.  The reliability is met for
all tasks not in~$I$, since they are executed at speed~\fr. It is also
met for all tasks in~$I$: $\forall i \in I,~ 1 - \lambda^2_0
\frac{x^2_i}{\fmin^2} \geq 1 - \lambda_0 \frac{w_i}{\fmax}$.  The
energy consumption is $E= \sum_{i\in I}2 a_i \fmin^2 + \sum_{i\notin
  I}a_i\fmax^2= 2X \fmin^2 + (S-X)\fmax^2 = E_0$. All bounds are
respected, and therefore the execution speeds are a solution
to~$\II_2$ (and each task keeps a constant speed during its whole
execution).

\medskip
\emph{Suppose now that $\II_2$ has a solution.}  
Since we consider the \VDD model, each execution 
can be run partly at speed~$\fmin$, and partly at speed~$\fmax$.
However, tasks executed only once are necessarily
executed only at maximum speed to match the reliability constraint. 

Let $I=\{i\; |\; T_i\mbox{ is executed twice in the solution}\}$. 
Let $Y=\sum_{i\in I}a_i$. We have $2Y = Y_1 + Y_2$, where $Y_1$ is the
total weight of each execution and re-execution ($2Y$) of tasks in $I$
that are executed at speed~$\fmin$, and $Y_2$ the total weight that is
executed at speed~$\fmax$.  We show that necessarily $Y_1 = 2X = 2Y$,
i.e., no part of any task in~$I$ is executed at
speed~\fmax. 

First let us show that $2X \leq 2Y$.
The energy consumption of the solution of~$\II_2$ is $E = Y_1\fmin^2+
Y_2\fmax^2 + (S-Y)\fmax^2 = Y_1\fmin^2+ (S - Y_1 + Y)\fmax^2$.  By
differentiating this function (with regards to $Y_1$, $E'= \fmin^2 -
\fmax^2 < 0$), we can see that the minimum is reached for $Y_1 = 2Y$
(since $Y_1 \in [0,2Y]$). 
Then, for $Y_1 = 2Y$, since the solution is such that $E \leq E_0$, we
have $E-E_0 = (Y-X)(2\fmin^2-\fmax^2) \leq 0$, and therefore $X \leq Y$.

Next let us show that $Y_1 \leq 2X$.  Suppose by contradiction that
$Y_1 > 2X$, then the execution time of the solution of~$\II_2$ is
$D=\frac{Y_1}{\fmin}+\frac{Y_2}{\fmax} + \frac{S-Y}{\fmax} =
\frac{Y_1}{\fmin}+\frac{S - Y_1 + Y}{\fmax} $. By differentiating this
function (with regards to $Y_1$), we can see it is strictly increasing
when $Y_1$ goes from $2X$ to $2Y$.  However, when $Y_1 = 2X+\epsilon$,
$D-D_0=\frac{\epsilon}{\fmin} + \frac{Y-X+\epsilon}{\fmax}>0$ (indeed,
each value of the sum is strictly positive). Hence, $Y_1 \leq 2X$. 

Finally, let us show that $Y_1 = 2X = 2Y$. Since $\II_2$ is a
solution, we know that $E \leq E_0$, and therefore $2X-Y_1 \geq
(Y+X - Y_1) \frac{\fmax^2}{\fmin^2} \geq (Y+X - Y_1)$  (the last
equality is only met when $Y+X - Y_1 =0$). Hence $2X \geq X+Y$, which is
only possible if $2X=X+Y$. This gives us the final result: $Y_1 = 2X =
2Y$ (all inequalities are tight).  

\smallskip
We conclude that $\sum_{i\in I}a_i = X$, and therefore $\II_1$~has
a solution. This concludes the proof.
\end{proof}

\medskip

In the following, we propose some polynomial time heuristics to tackle
the general tri-criteria problem. While these heuristics are designed for
the \continuous model, they can be easily adapted to the \VDD model
thanks to Proposition~\ref{vdd.2.speeds}.

\section{Heuristics for \tricritcont}
\label{heurWC}

In this section, building upon the theoretical results of
Section~\ref{cont_model}, we propose some polynomial time heuristics
for the \tricritcont problem, which was shown NP-hard (see
Theorem~\ref{cont_all_exec}). Recall that the mapping of the tasks
onto the processors is given, and we aim at reducing the energy
consumption by exploiting re-execution and speed scaling, while
meeting the deadline bound and all reliability constraints.

The first idea is inspired by 
Proposition~\ref{prop_WC_fr}: first we search for
the optimal solution of the problem instance without re-execution, a
phase that we call \emph{deceleration}: we slow down some tasks
if it can save energy without violating one of the constraints.  Then
we refine the schedule and choose the tasks that we want to
re-execute, according to some criteria.  We call \emph{type~A
  heuristics} such heuristics that obey this general scheme: first
deceleration then re-execution.  Type A heuristics are expected to be
efficient on a DAG with a low degree of parallelism (optimal for a
chain). 

However, Proposition~\ref{prop_WC_fork} (with fork graphs)
shows that it might be better to re-execute highly parallel tasks
before decelerating.  Therefore we introduce \emph{type~B heuristics},
which first choose the set of tasks to be re-executed, and then try to
slow down the tasks that could not be re-executed. We need to find
good criteria to select which tasks to re-execute, so that type B
heuristics prove efficient for DAGs with a high degree of parallelism.
In summary, type B heuristics obey the opposite scheme: first
re-execution then deceleration.

For both heuristic types, the approach for each phase can be sketched
as follows. Initially, each task is executed once at
speed~$\fmax$. Then, let $d_i$ be the finish time of task~$T_{i}$ in
the current configuration:
\begin{compactitem}
\item[{\em - Deceleration:}] We select a set of tasks that we execute
  at speed $\fdec = \max(\fr, \frac{\max_{i=1..n} d_i}{D}\fmax)$, which is the
  slowest possible speed meeting both the reliability and deadline
  constraints.
\item[{\em - Re-execution:}] We greedily select tasks for
  re-execution. The selection criterion is either by decreasing weights
  $w_{i}$, or by decreasing \emph{super-weights} $W_{i}$. The
  super-weight of a task $T_{i}$ is defined as the sum of the weights
  of the tasks (including $T_{i}$) whose execution interval is
  included into $T_{i}$'s execution interval. The rationale is that
  the super-weight of a task that we slow down is an estimation of
  the total amount of work that can be slowed down together with that
  task, hence of the energy potentially saved: this corresponds to the
  total slack that can be reclaimed.
\end{compactitem}

\noindent
We introduce further notations before listing the heuristics: 
\begin{compactitem}
\item[\emph{- SUS} (Slack-Usage-Sort)] is a function that sorts tasks by
  decreasing super-weights.
\item[\emph{- ReExec}] is a function that tries to re-execute the current
  task $T_{i}$, at speed $\freex = \frac{2c}{1+c} \fr$, where $c = 4
  \sqrt{\frac{2}{7}}
  \cos{\frac{1}{3}(\pi-\tan^{-1}{\frac{1}{\sqrt{7}}})}-1 ~ (\approx
  0.2838)$ (note that 
  $\freex$ is the optimal speed in the proof of
  Theorem~\ref{cont_all_exec}).  If it succeeds, it also re-executes
  at speed $\freex$ all the tasks that are taken into account to
  compute the super-weight of $T_{i}$. Otherwise, it does nothing.
\item[\emph{- ReExec\&SlowDown}] performs the same re-executions as
  \emph{ReExec} when it succeeds. But if the re-execution of the
  current task $T_{i}$ is not possible, it slows down $T_{i}$ as much
  as possible and does the same for all the tasks that are taken into
  account to compute the super-weight of~$T_{i}$.
\end{compactitem}

\smallskip
\noindent
We now detail the heuristics:

{\bf H\fmax.} In this heuristic, tasks are simply executed once at
maximum speed.

{\bf Hno-reex.}  In this heuristic, we do not allow any re-execution,
and we simply consider the possible deceleration of the tasks. We set
a uniform speed for all tasks, equal to $\fdec$, so that both the
reliability and deadline constraints are matched.  
Note that heuristics {\bf H\fmax} and {\bf Hno-reex} are identical except 
for a constant ratio on the speeds of each task, $\frac{\fmax}{\fdec}$.
Therefore, the energy ratio $\frac{E_{\text{\bf H\fmax}}}{E_{\text{\bf Hno-reex}}}$
is always equal to $\left (\frac{\fmax}{\fdec}\right )^2$ (for instance, 
if $\fmax=1$ and $\fdec=2/3$, then the energy ratio is equal to 2.25).

{\bf A.Greedy.}  This is a type A heuristic, where we first set the
speed of each task to~$\fdec$ (deceleration).  Let Greedy-List be the list
of all the tasks sorted according to decreasing weights~$w_i$.  Each
task~$T_i$ in Greedy-List is re-executed at speed~$\freex$ whenever
possible.  Finally, if there remains some slack at the end of the
processing, we slow down both executions of each re-executed task as
much as possible.

{\bf A.SUS-Crit.}  This is a type A heuristic, where we
first set the speed of each task to~$fdec$.
Let List-SW be the list of all tasks that belong to a critical path,
sorted according to SUS. We apply ReExec to List-SW (re-execution).
Finally we reclaim slack for re-executed tasks, similarly to the final
step of {\bf A.Greedy}.

{\bf B.Greedy.}  This is a type B heuristic. We use
Greedy-List as in heuristic {\bf A.Greedy}.  We try to re-execute
each task~$T_i$ of Greedy-List when possible.  Then, we slow down
both executions of each re-executed task~$T_i$ of Greedy-List as
much as possible.  Finally, we slow down the speed of each task of
Greedy-List that turn out not re-executed, as much as possible.

{\bf B.SUS-Crit.}  This is a type B heuristic. We use
List-SW as in heuristic {\bf A.SUS-Crit}.  We apply ReExec to List-SW
(re-execution).  Then we run Heuristic B.Greedy.

{\bf B.SUS-Crit-Slow.}  This is a type B heuristic. We use
List-SW, and we apply ReExec\&SlowDown (re-execution).  Then we use
Greedy-List: for each task~$T_i$ of Greedy-List, if there is
enough time, we execute twice $T_i$ at speed \freex (re-execution);
otherwise, we execute $T_{i}$ only once, at the slowest admissible
speed.

{\bf Best.}  This is simply the minimum value over the seven
previous heuristics, for reference.

The complexity of all these heuristics is bounded by O$(n^4\log n)$,
where $n$~is the number of tasks. The most time-consuming operation is
the computation of List-SW (the list of all elements belonging to a
critical path, sorted according to SUS).

\section{Simulations}
\label{sec.exp_results}

In this section, we report extensive simulations to assess the
performance of the heuristics presented in Section~\ref{heurWC}.  The
heuristics were coded in OCaml. The source code is publicly available
at~\cite{gaupy-web} (together with additional results that were
omitted due to lack of space).
	
\subsection{Simulation settings}

In order to evaluate the heuristics, we have generated DAGs using the
random DAG generation library GGEN~\cite{ggen}. Since GGEN does not
assign a weight to the tasks of the DAGs, we use a function that gives
a random float value in the interval $[0,10]$. Each simulation uses a
DAG with $100$ nodes and $300$ edges. We observe similar patterns for
other numbers of edges, see~\cite{gaupy-web} for further
information. 

We apply a critical-path list scheduling algorithm to map the
DAG onto the $p$ processors: we assign the most urgent ready task
(with largest bottom-level) to the first available processor.  The
bottom-level is defined as $bl(T_i)=w_i$ if $T_i$~has no successor
task, and $bl(T_i)=w_i + \max_{(T_i,T_j)\in \mathcal{E}} bl(T_j)$
otherwise.

We choose a reliability constant $\lambda_0 =
10^{-5}$~\cite{Assayad11} (we obtain identical results with other
values, see below).  Each reported result is the average on ten
different DAGs with the same number of nodes and edges, and the energy
consumption is normalized with the energy consumption returned by the
{\bf Hno-reex} heuristic. If the value is lower than~$1$, it means
that we have been able to save energy thanks to re-execution.

We analyze the influence of three different parameters: the tightness
of the deadline $D$, the number of processors $p$, and the reliability
speed $\fr$. In fact, the absolute deadline $D$ is irrelevant, and we
rather consider the \emph{deadline ratio} 
$\DDratio = \frac{D}{\DDmin}$, where \DDmin is the execution time when
executing each task once and at maximum speed~$\fmax$ (heuristic
H\fmax).  Intuitively, when the deadline ratio is close to~$1$, there
is almost no flexibility and it is difficult to re-execute tasks,
while when the deadline ratio is larger we expect to be able to slow
down and re-execute many tasks, thereby saving much more energy.

\subsection{Simulation results}

First note that with a single processor, heuristics A.SUS-Crit
and A.Greedy are identical, and heuristics B.SUS-Crit
and B.Greedy are identical (by definition, the only critical path 
is the whole set of tasks).

\smallskip
\noindent
\textbf{Deadline ratio.} 
In this set of simulations, we let $p \in \{1, 10, 50, 70\}$ and $\fr
= \frac{2}{3} \fmax$.  Figure~\ref{results_deadline} reports results
for $p=1$ and $p=50$.  When $p=1$, we see that the results are 
identical for all heuristics of type~A, and identical for all
heuristics of type~B.  
As expected from Proposition~\ref{prop_WC_fr}, type~A heuristics are better 
(see Figure~\ref{deadline_1_10}). 
With more processors (10, 50, 70), the results
have the same general shape: see Figure~\ref{deadline_50_10} with 50
processors.  When \DDratio is small, type B heuristics are better.
When \DDratio increases up to $1.5$, type A heuristics are closer to
type B ones. Finally, when \DDratio gets larger than $5$, all
heuristics converge towards the same result, where all tasks are
re-executed.

\begin{figure*}[h]
\begin{center}
\subfloat
{
    \includegraphics[scale=0.6]{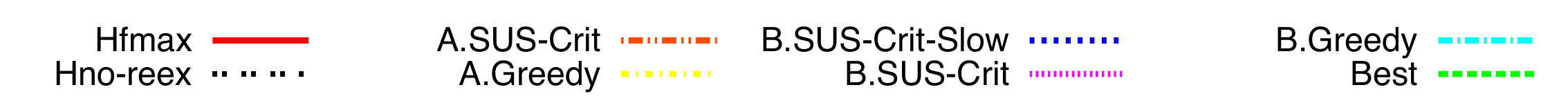}

}\\[.5cm]
\addtocounter{subfigure}{-1}
\subfloat[1 processor]
{\hspace{-1cm}
\includegraphics[scale=0.65]{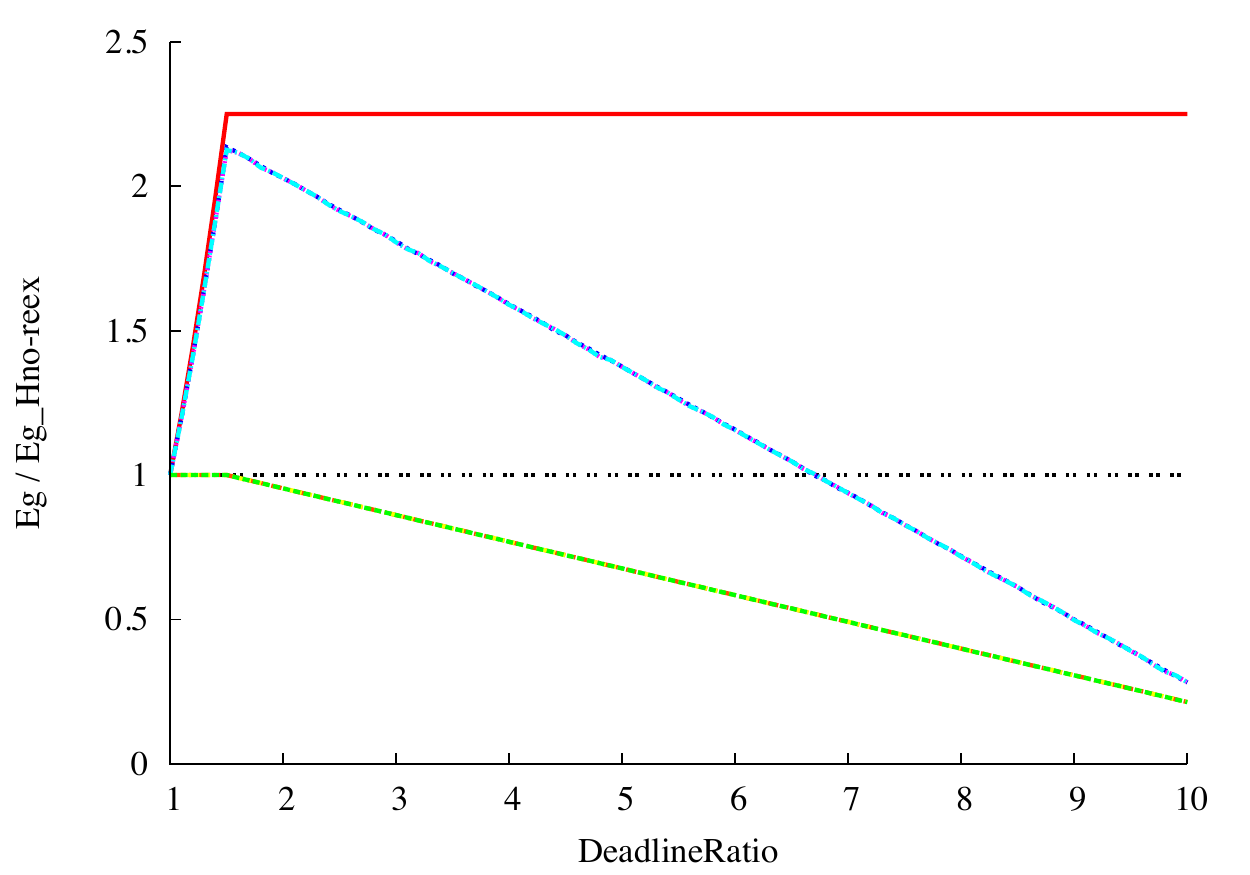}
	\label{deadline_1_10}
}
\subfloat[50 processors]
{
\includegraphics[scale=0.65]{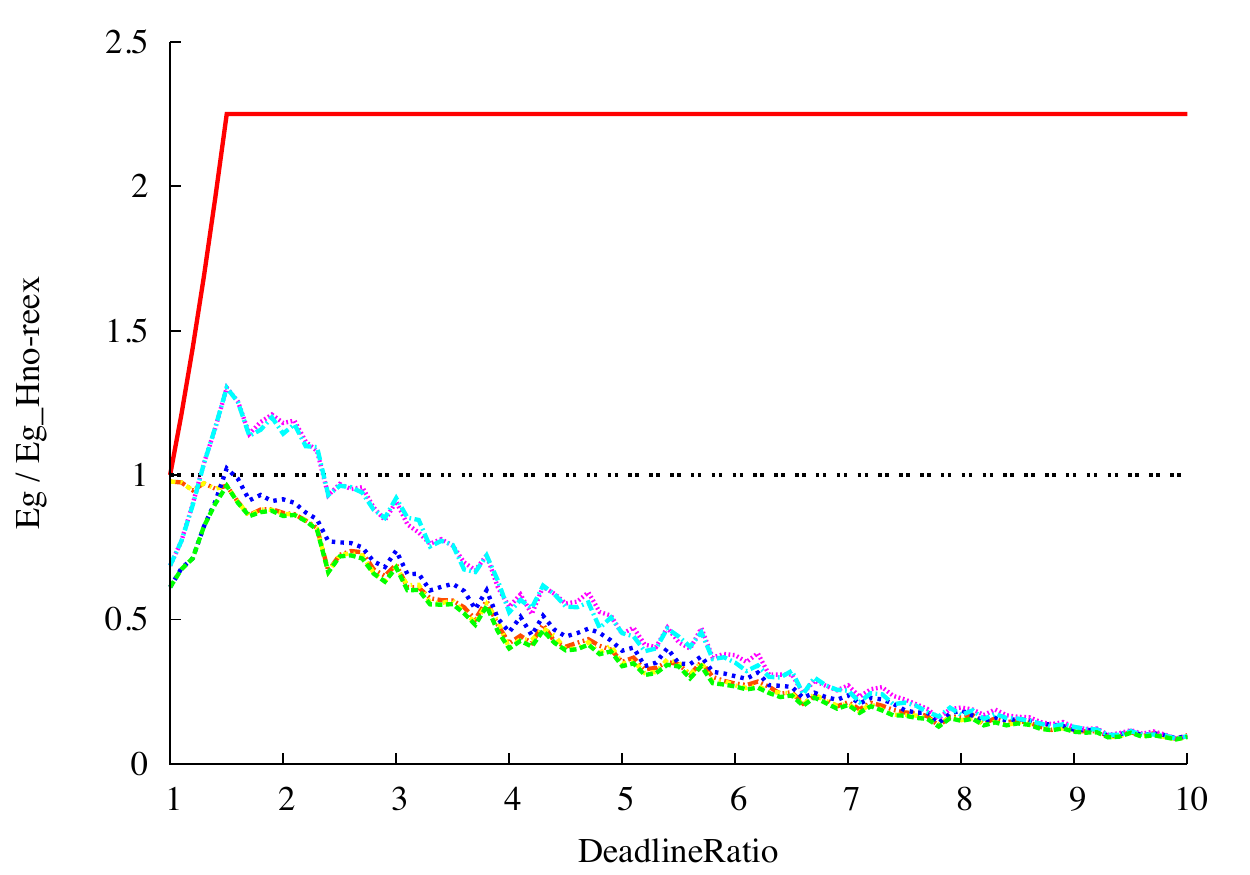}
	\label{deadline_50_10}
}	
\caption{Comparative study when the deadline ratio varies.}
\label{results_deadline}
\end{center}
\end{figure*}

\begin{figure*}[h]
\begin{center}
\subfloat[\DDratio = 1.2]
{\hspace{-1cm}
\includegraphics[scale=0.65]{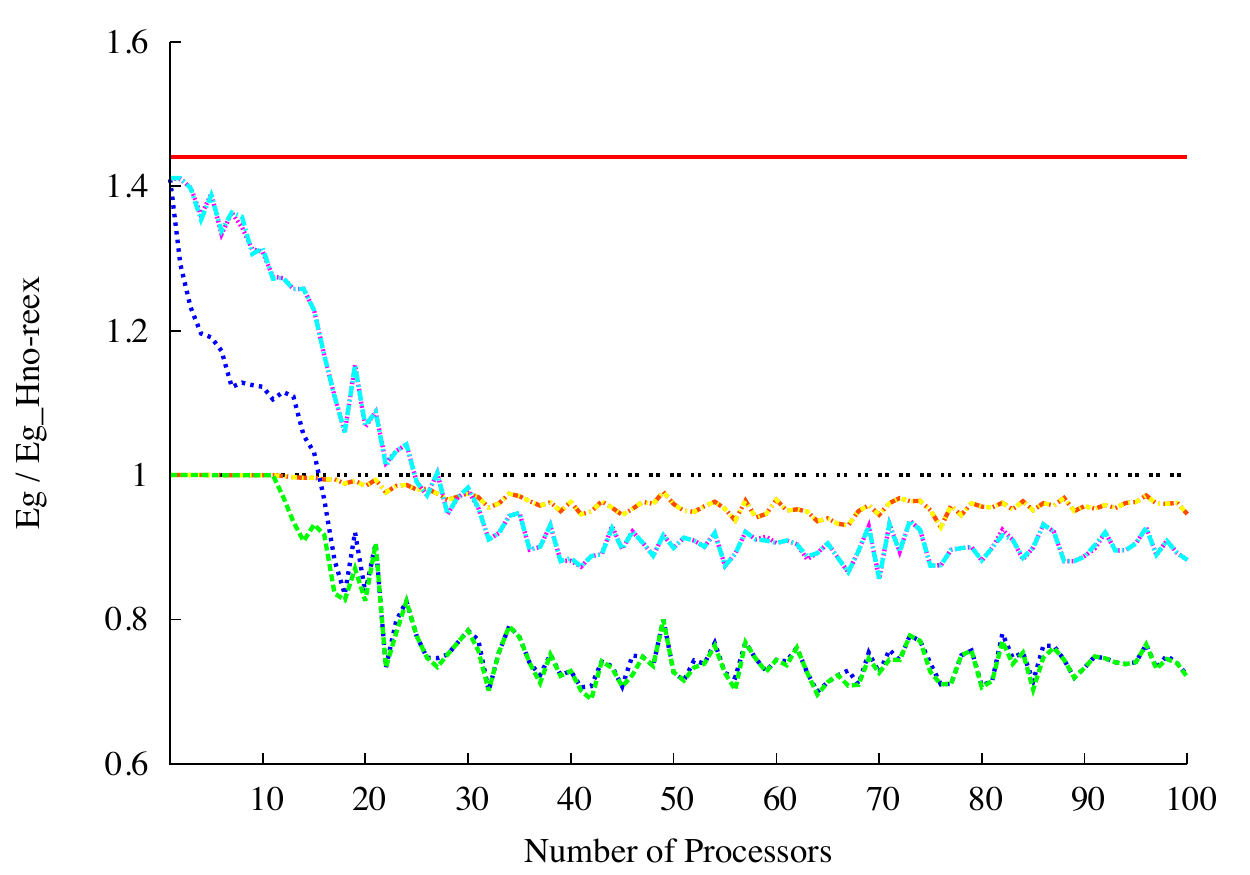}
	\label{proc_100_12}
}
\subfloat[\DDratio = 2.4]
{
\includegraphics[scale=0.65]{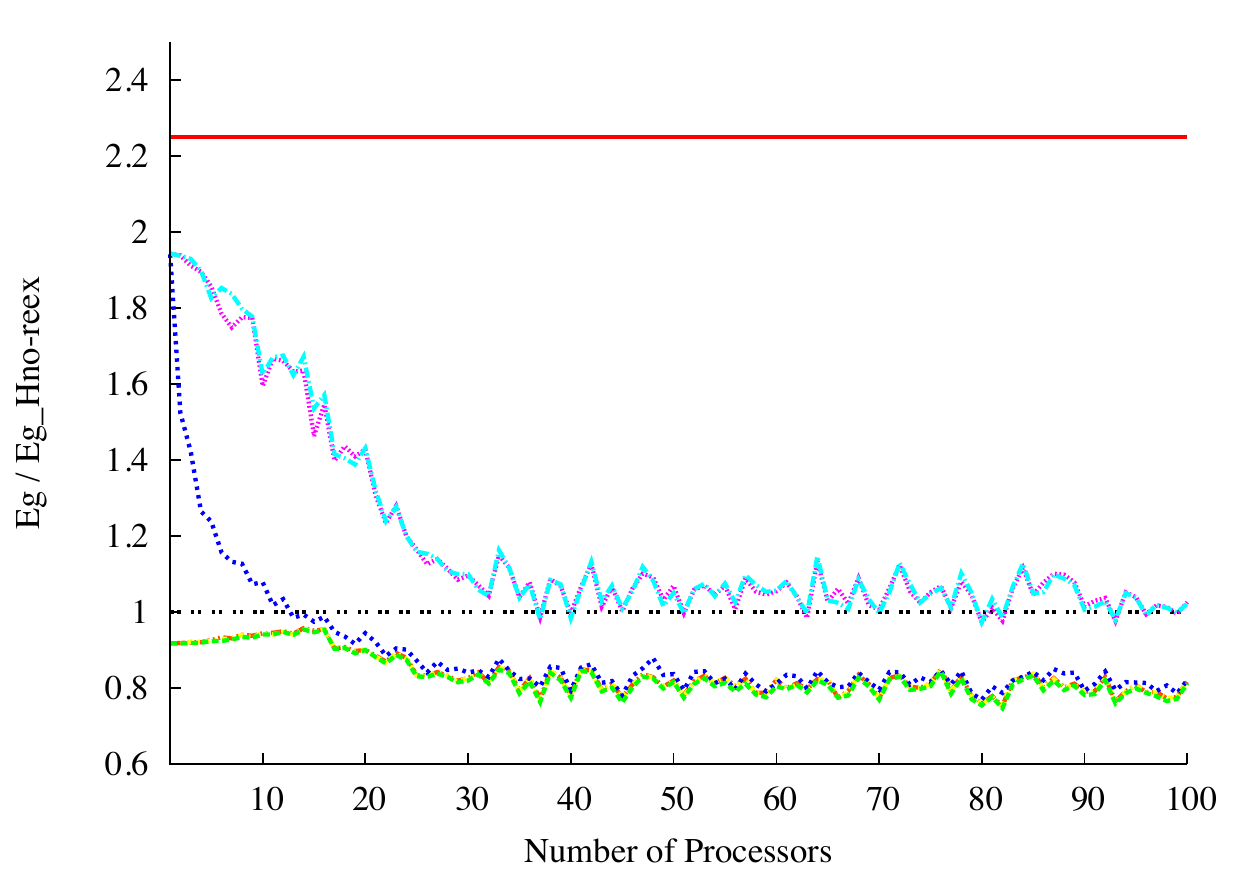}
	\label{proc_100_24}
}
\caption{Comparative study when the number of processors~$p$ varies.}
\vspace{-.5cm}
\label{results_proc}
\end{center}
\end{figure*}

\smallskip
\noindent
\textbf{Number of processors.}
In this set of simulations, we let $\DDratio \in \{1.2, 1.6, 2, 2.4\}$
and $\fr = \frac{2}{3} \fmax$.  Figure~\ref{results_proc} confirms
that type A heuristics are particularly efficient when the number of
processors is small, whereas type B heuristics are at their best when
the number of processors is large. Figure~\ref{proc_100_12} confirms the
superiority of type B heuristics for tight deadlines, as was observed
in Figure~\ref{deadline_50_10}.

\smallskip
\noindent
\textbf{Reliability \fr.}
In this set of simulations, we let $p \in \{1, 10, 50, 70\}$ and
$\DDratio \in \{1, 1.5, 3\}$. In Figure~\ref{results_fr}, there are
four different curves: the line at~$1$ corresponds to Hno-reex and
H\fmax, then come the heuristics of type A (that all obtain exactly
the same results), then B.SUS-Crit and B.Greedy that also obtain the
same results, and finally the best heuristic is B.SUS-Crit-Slow.  Note
that B.SUS-Crit and B.Greedy return the same results because they have
the same behavior when $\DDratio=1$: there is no liberty of action on
the critical paths. However B.SUS-Crit-Slow gives better results
because of the way it decelerates the important tasks that cannot be 
re-executed. 
When \DDratio is really
tight (equal to~$1$), decreasing the value of $\fr$ from $1$ 
to $0.9$ makes a real difference with type B heuristics.  We
observe an energy gain of 10\% when the number of processors is small
($10$ in Figure~\ref{fr_1_10_300}) and of 20\% with more processors
($50$ in Figure~\ref{fr_1_50_100}).

\smallskip
\noindent
\textbf{Reliability constant $\lambda_0$.} In
Figure~\ref{results_lambda}, we let $\lambda_0$ vary from $10^{-5}$ to
$10^{-6}$, and observe very similar results throughout this range of
values. Note that we did not plot H\fmax in this figure to ease the
readability.

\clearpage
\begin{figure*}[h]
\begin{center}
\subfloat
{
    \includegraphics[scale=0.6]{legende.pdf}

}\\[.5cm]
\addtocounter{subfigure}{-1}
\subfloat[\DDratio = 1;  $\;10$ processors]
{\hspace{-1cm}
\includegraphics[scale=0.65]{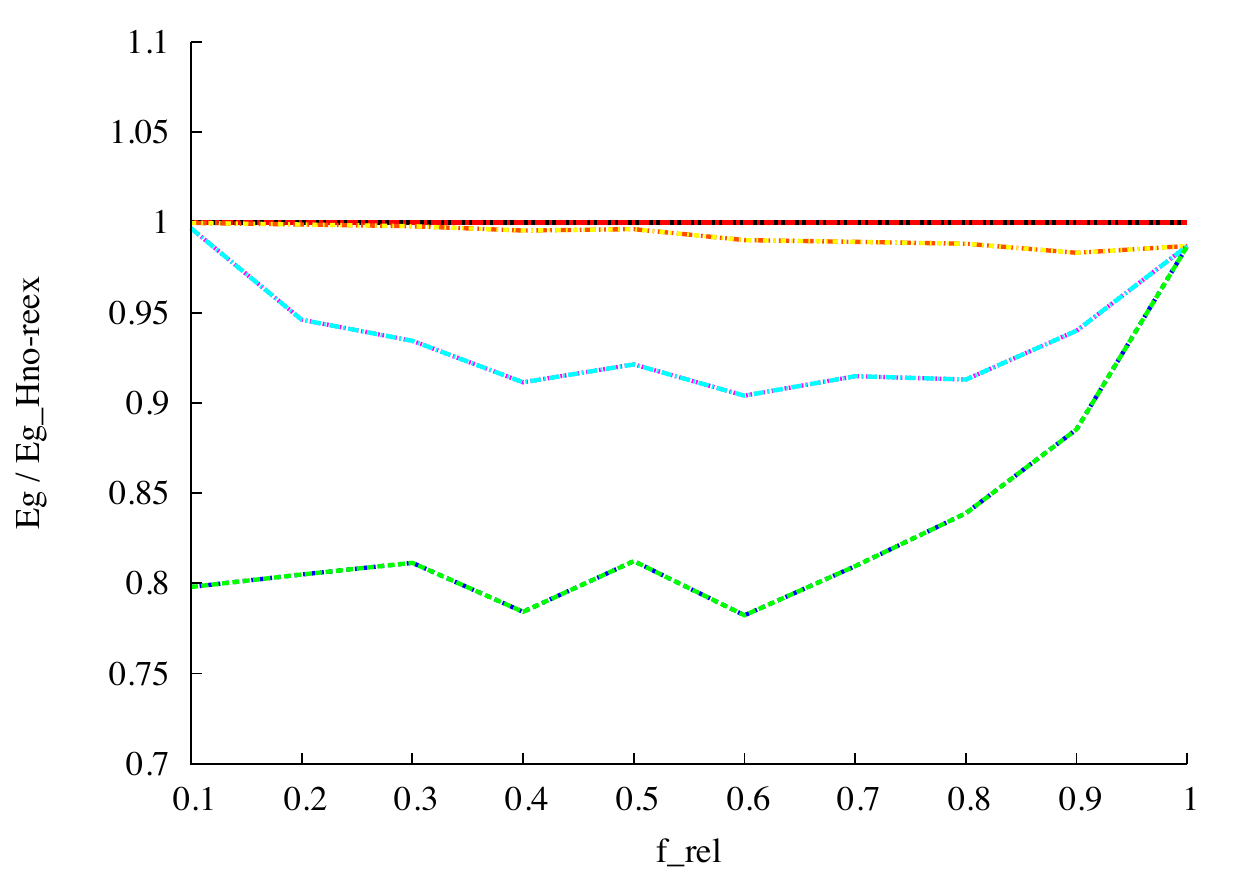}
	\label{fr_1_10_300}
}
\subfloat[\DDratio = 1;  $\;50$ processors]
{
\includegraphics[scale=0.65]{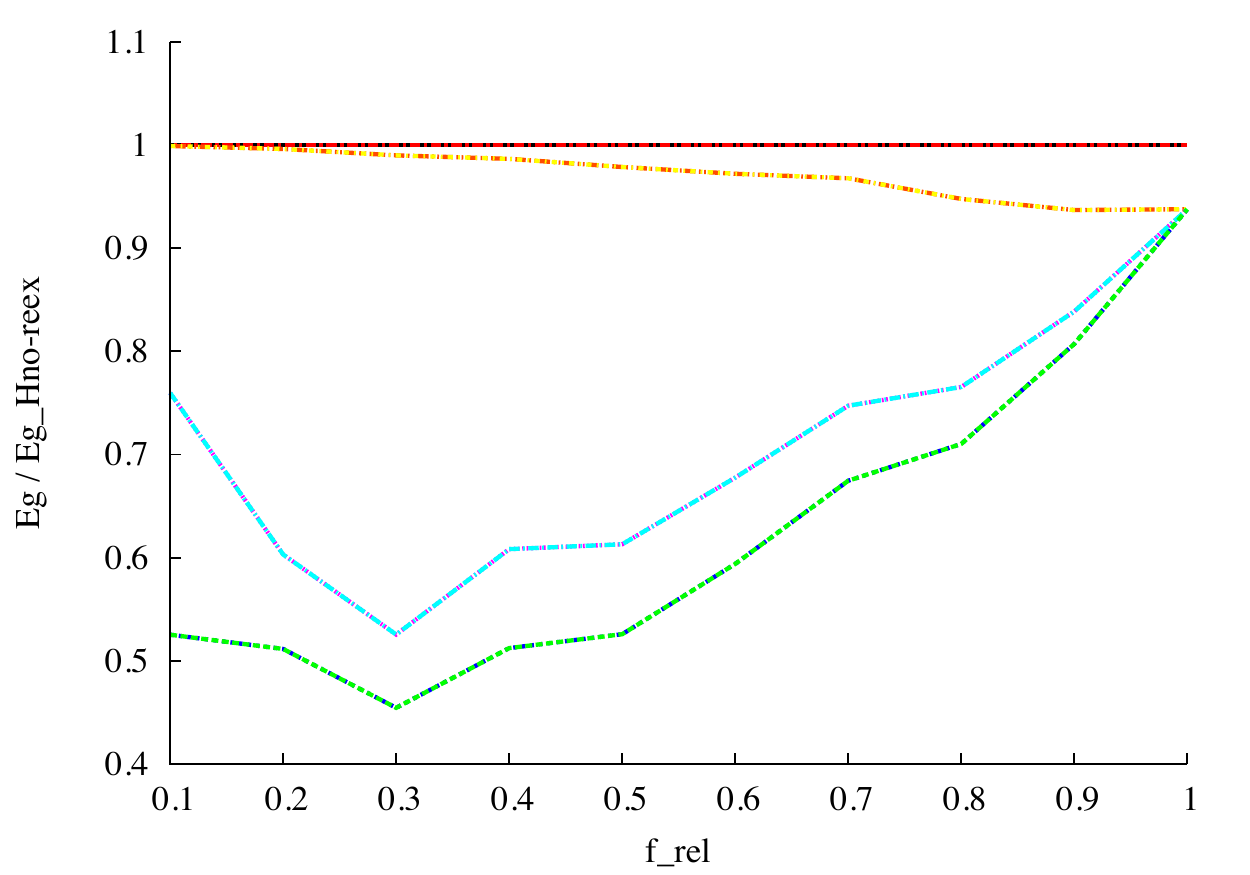}
	\label{fr_1_50_100}
}
\caption{Comparative study when the reliability \fr varies.}
\label{results_fr}
\end{center}
\end{figure*}

\begin{figure*}[h]
\begin{center}
\subfloat[\DDratio = 1.1;  $\;50$ processors]
{\hspace{-1cm}
\includegraphics[scale=0.65]{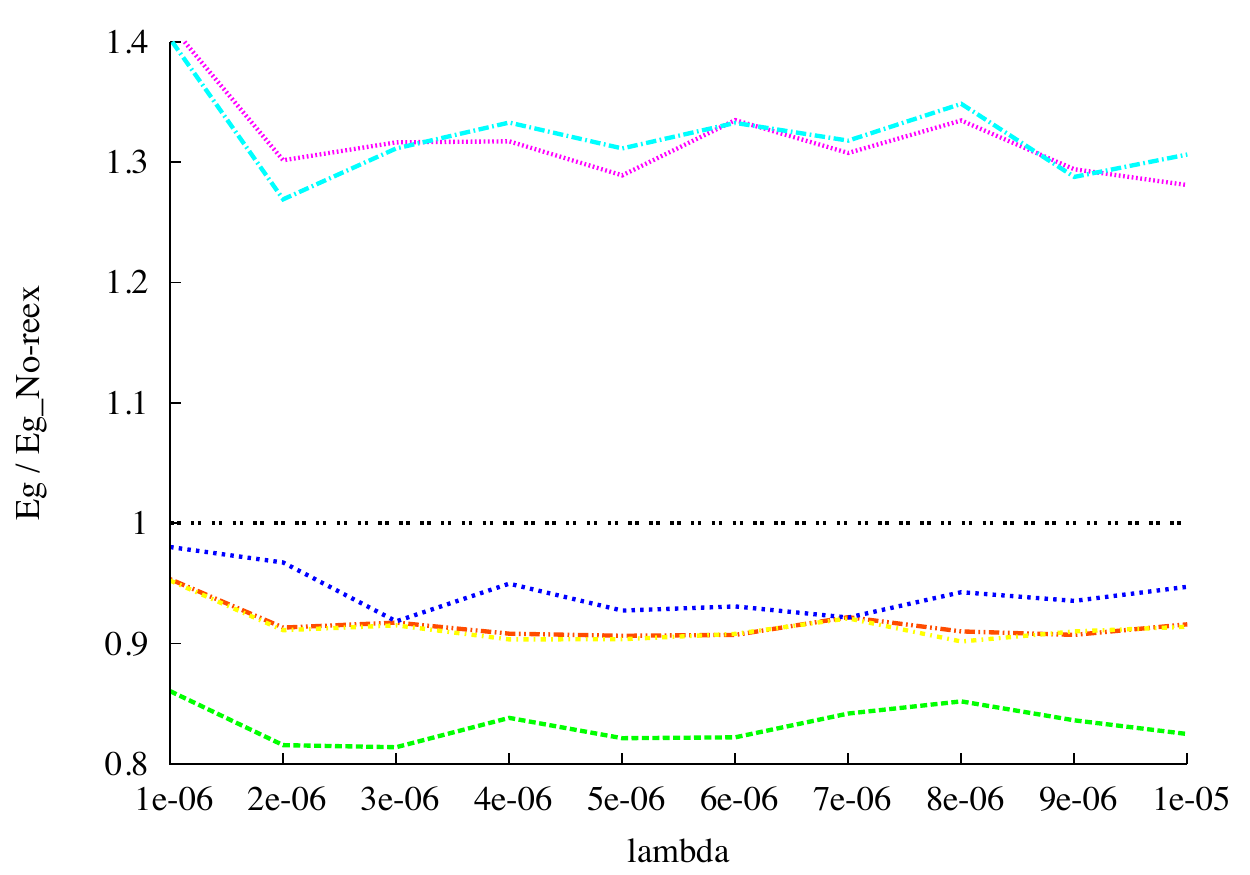}
	\label{lambda_1-1_50_300}
}
\subfloat[\DDratio = 1.5;  $\;50$ processors]
{
\includegraphics[scale=0.65]{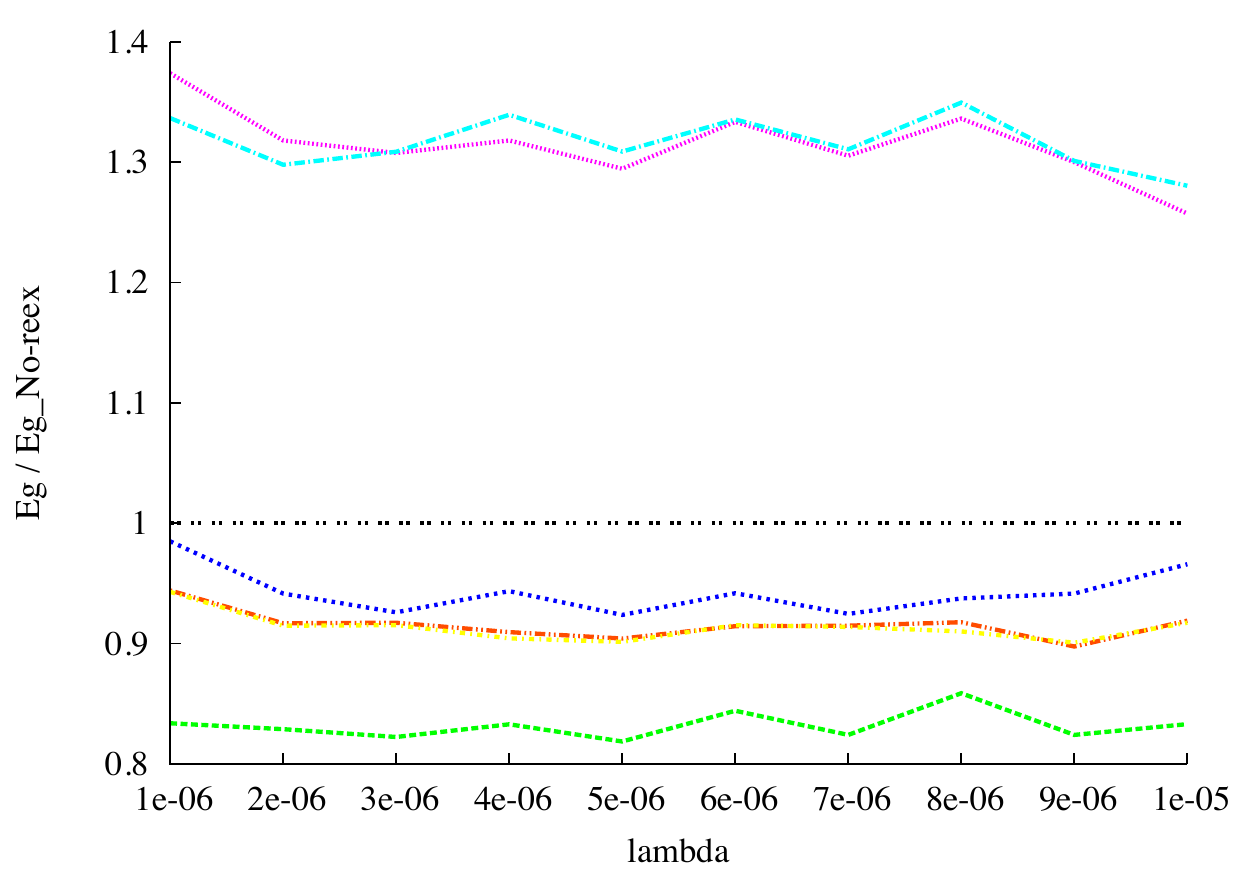}
	\label{lambda_1-5_50_300}
}
\caption{Comparative study when $\lambda_0$ varies.}
\vspace{-1cm}
\label{results_lambda}
\end{center}
\end{figure*}

\subsection{Understanding the results}

A.SUS-Crit and A.Greedy, and B.SUS-Crit and B.Greedy, often obtain
similar results, which might lead us to underestimate the importance
of critical path tasks.  However, the difference between
B.SUS-Crit-Slow and B.SUS-Crit shows otherwise.  Tasks that belong to
a critical path must be dealt with first.

A striking result is the impact of both the number of processors and
the deadline ratio on the effectiveness of the heuristics. Heuristics
of type~A, as suggested by Proposition~\ref{prop_WC_fr}, have much
better results when there is a small number of processors. When the
number of processors increases, there is a difference between small
and large deadline ratio. In particular, when the deadline ratio is
small, heuristics of type~B have better results. Indeed, heuristics of
type~A try to accommodate as many tasks as possible, and as a
consequence, no task can be re-executed. On the contrary, heuristics
of type~B try to favor some tasks that are considered as
important. This is highly profitable when the deadline is tight.

Note that all these heuristics take in average less than one {\em ms}
to execute on one instance, which is very reasonable. The heuristics
that compute the critical path (*.SUS-Crit-*) are the longest, and may
take up to two seconds when there are few processors. Indeed, the less
processors, the more edges there are in the dependence graph once the
task graph is mapped, and hence it increases the complexity of finding
the critical path. However, with more than ten processors, the running
time never exceeds two {\em ms}.

Altogether we have identified two very efficient and complementary
heuristics, A.SUS-Crit and B.SUS-Crit-Slow. Taking the best result out
of those two heuristics always gives the best result over all
simulations.

\section{Conclusion}
\label{sec.conclusion}

In this paper, we have accounted for the energy cost associated to
task re-execution in a more realistic and accurate way than the
best-case model used in~\cite{Zhu06}. Coupling this energy model with
the classical reliability model used in~\cite{Shatz89}, we have been
able to formulate a tri-criteria optimization problem: how to minimize
the energy consumed given a deadline bound and a reliability
constraint?  The ``antagonistic`` relation between speed and
reliability renders this tri-criteria problem much more challenging
than the standard bi-criteria (makespan, energy) version.  We
have stated two variants of the problem, for processor speeds obeying
either the \continuous or the \VDD model. We have assessed the
intractability of this tri-criteria problem, even in the case of a
single processor. In addition, we have provided several complexity
results for particular instances. 

We have designed and evaluated some polynomial-time heuristics for the
\tricritcont problem that are based on the failure probability, the
task weights, and the processor speeds.  These heuristics aim at
minimizing the energy consumption 
while enforcing reliability and deadline constraints. They rely on
\emph{dynamic voltage and frequency scaling} (DVFS) to decrease the
energy consumption. But because DVFS lowers the reliability of the
system, the heuristics use \emph{re-execution} to compensate for the
loss.  After running several heuristics on a wide class of problem
instances, we have identified two heuristics that are complementary,
and that together are able to produce good results on most
instances. The good news is that these results bring the first
efficient practical solutions to the tri-criteria optimization
problem, despite its theoretically challenging nature.
In addition, while the heuristics do not
modify the mapping of the application, it is possible to couple them
 with a list scheduling algorithm, as was done in the
simulations, in order to solve the more general problem in which the
mapping is not already given. 

\smallskip

Future work involves several promising directions.  On the theoretical
side, it would be very interesting to prove a competitive ratio for
the heuristic that takes the best out of A.SUS-Crit and
B.SUS-Crit-Slow. However, this is quite a challenging work for
arbitrary DAGs, and one may try to design approximation algorithms
only for special graph structures, e.g., series-parallel graphs.
Still, looking back at the complicated case analysis needed for an
elementary fork-graph with identical weights
(Proposition~\ref{prop_WC_fork}), we cannot underestimate the
difficulty of this problem.

While we have designed heuristics for the \tricritcont
model in this paper, we could easily adapt them to the \tricritvdd
model: for a solution given by a heuristic for \tricritcont, if a task
should be executed at the continuous speed~$f$, then we would execute
it at the two closest discrete speeds that bound~$f$, while matching
the execution time and reliability for this task. There remains to
quantify the performance loss incurred by the latter constraints.

Finally, we point out that energy reduction and reliability will be
even more important objectives with the advent of massively parallel
platforms, made of a large number of clusters of multi-cores.  More
efficient solutions to the tri-criteria optimization problem
(makespan, energy, reliability) could be achieved through combining
replication with re-execution.  A promising (and ambitious) research
direction would be to search for the best trade-offs that can be
achieved between these techniques that both increase reliability, but
whose impact on execution time and energy consumption is very
different.  We believe that the comprehensive set of theoretical
results and simulations given in this paper will provide solid
foundations for further studies, and constitute a partial yet
important first step for solving the problem at very large scale.


\noindent{\em Acknowledgments.} A.~Benoit and Y.~Robert are with the
Institut Universitaire de France.  This work was supported in part by
the ANR 
{\em RESCUE} project.

\bibliographystyle{abbrv}
\bibliography{biblio}

\begin{thebibliography}{10}

\bibitem{Assayad11}
I.~Assayad, A.~Girault, and H.~Kalla.
\newblock Tradeoff exploration between reliability power consumption and
  execution time.
\newblock In {\em Proc. of Conf. on Computer Safety, Reliability and Security
  (SAFECOMP)}, Washington, DC, USA, 2011. IEEE CS Press.

\bibitem{gaupy-web}
G.~Aupy.
\newblock Source code and data.
\newblock \url{http://gaupy.org/tri-criteria-scheduling}, 2012.

\bibitem{aupy12ccpe}
G.~Aupy, A.~Benoit, F.~Dufoss\'{e}, and Y.~Robert.
\newblock Reclaiming the energy of a schedule: models and algorithms.
\newblock {\em Concurrency and Computation: Practice and Experience}, 2012.
\newblock Also available as INRIA research report 7598 at \url{gaupy.org}.
  Short version appeared in SPAA'11.

\bibitem{pow3IPDPS}
H.~Aydin and Q.~Yang.
\newblock Energy-aware partitioning for multiprocessor real-time systems.
\newblock In {\em Proc. of Int. Parallel and Distributed Processing Symposium
  (IPDPS)}, pages 113--121. IEEE CS Press, 2003.

\bibitem{Baleani03}
M.~Baleani, A.~Ferrari, L.~Mangeruca, A.~Sangiovanni-Vincentelli, M.~Peri, and
  S.~Pezzini.
\newblock Fault-tolerant platforms for automotive safety-critical applications.
\newblock In {\em Proc. of Int. Conf. on Compilers, Architectures and Synthesis
  for Embedded Systems}, pages 170--177. ACM Press, 2003.

\bibitem{BKP07}
N.~Bansal, T.~Kimbrel, and K.~Pruhs.
\newblock Speed scaling to manage energy and temperature.
\newblock {\em Journal of the ACM}, 54(1):1 -- 39, 2007.

\bibitem{Brucker}
P.~Brucker.
\newblock {\em Scheduling Algorithms}.
\newblock Springer, 2007.

\bibitem{pow3ICPP}
J.-J. Chen and T.-W. Kuo.
\newblock Multiprocessor energy-efficient scheduling for real-time tasks.
\newblock In {\em Proc. of Int. Conf. on Parallel Processing (ICPP)}, pages
  13--20. IEEE CS Press, 2005.

\bibitem{ggen}
D.~Cordeiro, G.~Mounié, S.~Perarnau, D.~Trystram, J.-M. Vincent, and
  F.~Wagner.
\newblock Random graph generation for scheduling simulations.
\newblock In {\em Proc. of 3rd Int. {ICST} Conf. on Simulation Tools and
  Techniques ({SIMUT}ools 2010)}, page~10, mar 2010.

\bibitem{Degal05SEI}
V.~Degalahal, L.~Li, V.~Narayanan, M.~Kandemir, and M.~J. Irwin.
\newblock Soft errors issues in low-power caches.
\newblock {\em IEEE Trans. Very Large Scale Integr. Syst.}, 13:1157--1166,
  October 2005.

\bibitem{GareyJohnson}
M.~R. Garey and D.~S. Johnson.
\newblock {\em Computers and Intractability; A Guide to the Theory of
  NP-Completeness}.
\newblock W. H. Freeman \& Co., New York, NY, USA, 1990.

\bibitem{Girault09}
A.~Girault, E.~Saule, and D.~Trystram.
\newblock Reliability versus performance for critical applications.
\newblock {\em J. Parallel Distrib. Comput.}, 69:326--336, March 2009.

\bibitem{LeeSakurai00}
S.~Lee and T.~Sakurai.
\newblock Run-time voltage hopping for low-power real-time systems.
\newblock In {\em Proc. of Annual Design Automation Conf. (DAC)}, pages
  806--809, 2000.

\bibitem{Melhem03CP}
R.~Melhem, D.~Mosse, and E.~Elnozahy.
\newblock The interplay of power management and fault recovery in real-time
  systems.
\newblock {\em IEEE Trans. on Computers}, 53:2004, 2003.

\bibitem{Miermont2007}
S.~Miermont, P.~Vivet, and M.~Renaudin.
\newblock {A Power Supply Selector for Energy- and Area-Efficient Local Dynamic
  Voltage Scaling}.
\newblock In {\em Integrated Circuit and System Design. Power and Timing
  Modeling, Optimization and Simulation}, volume 4644, pages 556--565. Springer
  Berlin / Heidelberg, 2007.

\bibitem{Oliner04}
A.~J. Oliner, R.~K. Sahoo, J.~E. Moreira, M.~Gupta, and A.~Sivasubramaniam.
\newblock {Fault-aware job scheduling for BlueGene/L systems}.
\newblock In {\em Proc. of Int. Parallel and Distributed Processing Symposium
  (IPDPS)}, pages 64--73, 2004.

\bibitem{Izo07}
P.~Pop, K.~H. Poulsen, V.~Izosimov, and P.~Eles.
\newblock Scheduling and voltage scaling for energy/reliability trade-offs in
  fault-tolerant time-triggered embedded systems.
\newblock In {\em Proc. of IEEE/ACM Int. Conf. on Hardware/software codesign
  and system synthesis (CODES+ISSS)}, pages 233--238, 2007.

\bibitem{prathipati2004}
R.~B. Prathipati.
\newblock Energy efficient scheduling techniques for real-time embedded
  systems.
\newblock Master's thesis, Texas A\&M University, May 2004.

\bibitem{Rayward95}
V.~J. Rayward-Smith, F.~W. Burton, and G.~J. Janacek.
\newblock Scheduling parallel programs assuming preallocation.
\newblock In P.~Chr\'etienne, E.~G. {Coffman Jr.}, J.~K. Lenstra, and Z.~Liu,
  editors, {\em Scheduling Theory and its Applications}. John Wiley and Sons,
  1995.

\bibitem{rudin-principles}
W.~Rudin.
\newblock {\em Principles of mathematical analysis}.
\newblock McGraw-Hill Book Co., New York, third edition, 1976.
\newblock International Series in Pure and Applied Mathematics.

\bibitem{Shatz89}
S.~M. Shatz and J.-P. Wang.
\newblock Models and algorithms for reliability-oriented task-allocation in
  redundant distributed-computer systems.
\newblock {\em IEEE Transactions on Reliability}, 38:16--27, 1989.

\bibitem{Wang2010}
L.~Wang, G.~von Laszewski, J.~Dayal, and F.~Wang.
\newblock {Towards Energy Aware Scheduling for Precedence Constrained Parallel
  Tasks in a Cluster with DVFS}.
\newblock In {\em Proc. of CCGrid'2010, the 10th IEEE/ACM Int. Conf. on
  Cluster, Cloud and Grid Computing}, pages 368 --377, May 2010.

\bibitem{Zhang03CP}
Y.~Zhang and K.~Chakrabarty.
\newblock Energy-aware adaptive checkpointing in embedded real-time systems.
\newblock In {\em Proc. of Conf. on Design, Automation and Test in Europe
  (DATE)}, page 10918. IEEE CS Press, 2003.

\bibitem{Zhu06}
D.~Zhu and H.~Aydin.
\newblock Energy management for real-time embedded systems with reliability
  requirements.
\newblock In {\em Proc. of IEEE/ACM Int. Conf. on Computer-Aided Design
  (ICCAD)}, pages 528--534, 2006.

\bibitem{Zhu04EEM}
D.~Zhu, R.~Melhem, and D.~Moss\'e.
\newblock The effects of energy management on reliability in real-time embedded
  systems.
\newblock In {\em Proc. of IEEE/ACM Int. Conf. on Computer-Aided Design
  (ICCAD)}, pages 35--40, Washington, DC, USA, 2004. IEEE CS Press.

\end{thebibliography}

\end{document}